\documentclass[11pt]{article}

\usepackage[margin=1.13in]{geometry}

\usepackage{graphicx}
\usepackage[caption=false]{subfig}
\usepackage{amsmath, amssymb, amsthm}
\usepackage[colorlinks=true, urlcolor=blue, citecolor=blue, linkcolor=blue, filecolor=blue]{hyperref}

\usepackage{comment}
\usepackage{soul} 
\usepackage{sidecap}

\newcommand{\pathToCommon}{.}
\newcommand{\pathToCommonFigs}{.}
\newcommand{\R}{\mathbb{R}}
\newcommand{\Z}{\mathbb{Z}}
\newcommand{\ep}{\varepsilon}
\renewcommand{\Pr}[1]{\mathbf{P} \left( {#1} \right)}
\newcommand{\E}[1]{\mathbf{E} \left( {#1} \right)}
\renewcommand{\O}[1]{O \left( {#1} \right) }

\renewcommand{\P}{\boldsymbol{\mathcal{P}}}
\newcommand{\Q}{\boldsymbol{\mathcal{Q}}}
\newcommand{\F}{\boldsymbol{\mathcal{F}}}
\newcommand{\mcZ}{\mathcal{Z}}
\newcommand{\hvphi}{{\hat{\varphi}}}
\newcommand{\reliable}[2]{ {R^{#1}(#2)} }
\newcommand{\reliablestat}[3]{ {R^{#2}_{#1}(#3)} }
\newcommand{\interpretable}[2]{ {I^{#1}(#2)} }
\newcommand{\interpretablestat}[3]{ {I^{#2}_{#1}(#3)} }
\newcommand{\powerfunc}[2]{ {K^{#2}_{#1}} }

\newcommand{\MIC}{\textnormal{MIC}}
\newcommand{\popMIC}{{\mbox{MIC}_*}}
\newcommand{\MICestE}{{\mbox{MIC}_e}}
\newcommand{\MICestD}{{\mbox{MIC}_d}}

\newcommand{\proofwaypoint}[1]{ {\em {#1}:} }

\newtheorem{thm}{Theorem}[section]
\newtheorem*{thm*}{Theorem}
\newtheorem{prop}[thm]{Proposition}
\newtheorem{lemma}[thm]{Lemma}
\newtheorem{cor}[thm]{Corollary}
\newtheorem{rmk}[thm]{Remark}

\theoremstyle{definition}
\newtheorem{definition}[thm]{Definition}

\title{Theoretical Foundations of Equitability\\and the Maximal Information Coefficient \footnote{This manuscript is subsumed by \cite{reshef2015equitability} and \cite{reshef2015estimating}. Please cite those papers instead.}}
\author{
Yakir A. Reshef\thanks{School of Engineering and Applied Sciences, Harvard University. \href{mailto:yakir@seas.harvard.edu}{\nolinkurl{yakir@seas.harvard.edu}}.}
\and
David N. Reshef\thanks{Department of Electrical Engineering and Computer Science, Massachusetts Institute of Technology. \href{mailto:dnreshef@mit.edu}{\nolinkurl{dnreshef@mit.edu}}.}
\and
Pardis C. Sabeti\thanks{Department of Organismic and Evolutionary Biology, Harvard University. \href{psabeti@oeb.harvard.edu}{\nolinkurl{psabeti@oeb.harvard.edu}}.}
\and
Michael Mitzenmacher\thanks{School of Engineering and Applied Sciences, Harvard University. \href{michaelm@eecs.harvard.edu}{\nolinkurl{michaelm@eecs.harvard.edu}}.}
}

\begin{document}

\maketitle
\begin{abstract}
The maximal information coefficient ($\MIC$) is a tool for finding the strongest pairwise relationships in a data set with many variables \cite{MINE}. $\MIC$ is useful because it gives similar scores to equally noisy relationships of different types. This property, called {\em equitability}, is important for analyzing high-dimensional data sets.

Here we formalize the theory behind both equitability and $\MIC$ in the language of estimation theory. This formalization has a number of advantages. First, it allows us to show that equitability is a generalization of power against statistical independence. Second, it allows us to compute and discuss the population value of $\MIC$, which we call $\popMIC$. In doing so we generalize and strengthen the mathematical results proven in \cite{MINE} and clarify the relationship between $\MIC$ and mutual information. Introducing $\popMIC$ also enables us to reason about the properties of $\MIC$ more abstractly: for instance, we show that $\popMIC$ is continuous and that there is a sense in which it is a canonical ``smoothing" of mutual information. We also prove an alternate, equivalent characterization of $\popMIC$ that we use to state new estimators of it as well as an algorithm for explicitly computing it when the joint probability density function of a pair of random variables is known. Our hope is that this paper provides a richer theoretical foundation for $\MIC$ and equitability going forward.

This paper will be accompanied by a forthcoming companion paper that performs extensive empirical analysis and comparison to other methods and discusses the practical aspects of both equitability and the use of $\MIC$ and its related statistics.
\end{abstract}

\section{Introduction}

Suppose we have a data set with hundreds or thousands of variables, and we wish to find the strongest pairwise associations in that data set. The number of pairs of will be in the hundreds of thousands, or even millions, and so manually examining each pairwise scatter plot is out of the question. In such a context, one commonly taken approach is to compute some statistic on each of these pairs of variables, to rank the variable pairs from highest- to lowest-scoring, and then to examine only the top of the resulting list.

The results of this approach depend heavily on the chosen statistic. In particular, suppose the statistic is a measure of dependence, meaning that its population value is non-zero score exactly in cases of statistical dependence. Even with such a guarantee, the magnitude of the this non-zero score may depend heavily on the type of dependence in question, thereby skewing the top of the list toward certain types of relationships over others. For instance, a statistic may give non-zero scores to both linear and sinusoidal relationships; however, if the scores of the linear relationships are systematically higher, then using this statistic to rank variable pairs in a large data set will cause the many linear relationships in the data set to crowd out any potential sinusoidal relationships from the top of the list. This means that the human examining the top of the list will effectively never see the sinusoidal relationships.

This shortcoming is not as concerning in a hypothesis testing framework: if what we sought were a comprehensive list of all the non-trivial associations in a data set, then all we would care about would be that the sinusoidal relationships are {\em detected} with sufficient power so that we could reject the null hypothesis. Many excellent methods exist that allow one to test for independence in this way in various settings \cite{szekely2009brownian, renyi1959measures, bryc2002maxcorr, lopez2013randomized, gretton2005measuring, Kraskov, heller2013consistent}. However, often in data exploration the goal is to identify a relatively small set of strongest associations within a dataset as opposed to finding as many non-trivial associations as possible, which often are too many to sift through. What is needed then is a measure of dependence whose values, in addition to allowing us to {\em identify} significant relationships (i.e. reject a null hypothesis of independence), also allow us to measure the {\em strength} of relationships (i.e. estimate an effect size).

With the goal of addressing this need, we introduced in ~\cite{MINE} a notion called {\em equitability}: an equitable measure of dependence is one that assigns similar scores to relationships with equal noise levels, regardless of relationship type. This notion is notably imprecise -- it does not specify, for example, which relationship types are covered nor what is meant by ``noise" or ``similar". However, we noted that in the case of functional relationships, one reasonable definition of equitability might be that the value of the statistic reflect the coefficient of determination ($R^2$) of the data with respect to the regression function with as weak a dependence as possible on the function in question. Additionally, though characterizing noise in the case of superpositions of several functional relationships is difficult, it seems reasonable to require that the statistic give a perfect score when the functional relationships being superimposed are noiseless. We then introduced a statistic, the maximal information coefficient ($\MIC$), that behaves more equitably on functional relationships than the state of the art and also has the desired behavior on superpositions of functional relationships, given sufficient sample size.

Although $\MIC$ has enjoyed widespread use in a variety of disciplines \cite{koren2012host, maurice2013xenobiotics, qu2012bovine, bindea2013cluepedia, malod2012characterizing, das2012genome, anderson2012ranking, lin2012maximal, riccadonna2012dtw, sagl2012ubiquitous}, the original paper on equitability and $\MIC$ has generated much discussion, both published and otherwise, including some concerns and confusions. Perhaps the most frequent concern that we have heard is the desire for a richer and more formal theoretical framework for equitability, and this is the main issue we address in this paper.  In particular, we provide a formal definition for equitability that is sufficiently general to allow us to state in one unified language several of the variants of the original concept that have arisen.  We use this language to discuss the result of Kinney and Atwal about the impossibility of perfect equitability in some settings \cite{kinney2014equitability}, explaining its limitations based on the permissive underlying noise model that it requires, and on the strong assumption of perfect equitability that it makes. (See Section~\ref{subsubsec:kinney}.) We also use the formal definition of equitability to clarify the relationship between equitability and statistical power by proving an equivalent characterization of equitability in terms of power against a certain set of null hypotheses. Specifically, we show that whereas typical measures of dependence are analyzed in terms of the power of their corresponding tests to distinguish statistical independence from non-trivial associations, an equitable statistic is one that yields tests that can distinguish finely between relationships of different strengths that may {\em both} be non-trivial. We then explain how this relates to raised concerns about the power of $\MIC$ \cite{simon2012comment}. (See Section~\ref{sec:equitAndPower}.)

Following our treatment of equitability, we show that $\MIC$ can be viewed as a consistent estimator of a population quantity, which we call $\popMIC$, and give a closed-form expression for it. This has several benefits. First, the consistency of $\MIC$ as an estimator of $\popMIC$, together with properties of $\popMIC$ that are easy to prove given the closed-form expression, trivially subsumes and generalizes many of the theorems proven in \cite{MINE} about the properties of $\MIC$. Second, it clarifies that the parameter choices in $\MIC$ are not fundamental to the definition of the estimand ($\popMIC$) but rather simply control the bias and variance of the estimator ($\MIC$). Third, separating finite-sample effects from properties of the estimand $\popMIC$ allows us to rigorously discuss the theoretical relationship between $\popMIC$ and mutual information. And finally, since power is a property of the test corresponding to a statistic at finite sample sizes and not of the population value of the statistic, this re-orientation allows us to ask whether there exist estimates of $\popMIC$ other than $\MIC$ that retain the relative equitability of $\MIC$ but also result in better power against statistical independence. It turns out that there do, as we shall soon discuss.

Having a closed-form expression the population value of $\MIC$ (i.e. $\popMIC$) also allows us to reason about it more abstractly, and this is the goal to which we devote the remainder of the paper. We first show that, considered as a function of probability density functions, $\popMIC$ is continuous. This further clarifies the relationship with mutual information by allowing us to view $\popMIC$ as a ``minimally smoothed" version of mutual information that is uniformly continuous. (In contrast, mutual information alone is not continuous in this sense.)

Our theory also yields an equivalent characterization of $\popMIC$ that allows us to develop better estimators of it. The expression for $\popMIC$ given at the beginning of this paper, which is analogous to the expression for $\MIC$, defines it as the supremum of a matrix called the {\em characteristic matrix}. We show here that $\popMIC$ can instead be viewed as the supremum only of the {\em boundary} of that matrix. This is theoretically interesting, but it is also practically important, because computing elements of this boundary is easier than computing elements of the original matrix. In particular, our equivalent characterization of $\popMIC$ leads to the following advances.

\begin{itemize}
\item A new consistent estimator of $\popMIC$, which we call $\MICestE$, and an exact, efficient algorithm for computing it. (The algorithm introduced in \cite{MINE} for computing $\MIC$ is only a heuristic.)
\item An approximation algorithm for computing the $\popMIC$ of a given probability density function. Previously only heuristic algorithms were known \cite{MINE, lee2013resolution}. Having such an algorithm enables the evaluation of different estimators of $\popMIC$ as well as the evaluation of properties of $\popMIC$ in the infinite-data limit.
\item An estimator of $\popMIC$, which we call $\MICestD$, that proceeds by using a consistent density estimator to estimate the probability density function that gave rise to the observed samples, and then applying the above algorithm to compute the true $\popMIC$ of the resulting probability density function. This approach may prove more accurate in cases where we can encode into our density estimator some prior knowledge about the types of distributions we expect.
\end{itemize}

This paper will be accompanied by a forthcoming companion paper that performs extensive empirical analysis of several methods, including comparisons of $\MIC$ and related statistics. Among other things, that paper show that the first of the two estimators introduced here, $\MICestE$, yields a significant improvement in terms of equitability, power, bias/variance properties, and runtime over the original statistic introduced in~\cite{MINE}. The companion paper also compares both $\MICestE$ as well as the original statistic from~\cite{MINE} to existing methods along these same criteria and discusses when equability is a useful desideratum for data exploration in practice. Thus, questions concerning the performance of $\MIC$ (as well as other estimates of $\popMIC$) compared to other methods at finite sample sizes are deferred to the companion paper, while this paper focuses on theoretical issues. Our hope is that these two papers together will lay a richer theoretical foundation on which others can build to improve our knowledge of equitability and $\popMIC$, and to advance our understanding of when equitability and estimation of $\popMIC$ are useful in practice.

\section{Equitability}
\label{sec:equitability}
Equitability has been described informally as the ability of a statistic to ``give similar scores to equally noisy relationships of different types" \cite{MINE}. Here we provide the formalism necessary to discuss this notion more rigorously and define equitability using the language of estimation theory. Although what we are ultimately interested in is the equitability of a statistic, we first define equitability and discuss variations on the definition in the setting of random variables. Only then do we adapt our definitions to incorporate the uncertainty that comes with working with finite samples rather than random variables.

\subsection{Overview}
\label{subsec:equitOverview}
Before we formally define equitability in full generality, we first give a semi-formal overview of how we will do so, as well as a brief discussion of the benefits of our approach.

In \cite{MINE}, we asked to what extent evaluating a statistic like $\MIC$ on a sample from a noisy functional relationship with joint distribution $\mcZ$ tells us about the coefficient of determination ($R^2$) of that relationship\footnote{Recall that for a pair of jointly distributed random variables $(X, Y)$ with regression function $f$, $R^2$ is the squared Pearson correlation coefficient between $f(X)$ and $Y$}. However, this setup can be generalized as follows. We have a statistic $\hvphi$ (e.g. $\MIC$) that detects any deviation from statistical independence, a set $\Q$ of distinguished {\em standard relationships} on which we are able to define what we mean by noise (e.g. noisy functional relationships), and a {\em property of interest} $\Phi : \Q \rightarrow [0,1]$ that quantifies the noise in those relationships (e.g. $R^2$). We now ask: to what extent will evaluating $\hvphi$ on a sample from some joint distribution $\mcZ \in \Q$ tell us about $\Phi(\mcZ)$?

How will we quantify this? Let us step back for a moment and suppose that finite sample effects are not an issue: we will consider the population value $\varphi$ of $\hvphi$ and discuss its desired behavior on distributions $\mcZ$. In this setting, we will want $\varphi$ to have the following two properties.
\begin{enumerate}
\item It is a measure of dependence. That is, $\varphi(\mcZ) = 0$ if and only if $\mcZ$ exhibits statistical independence.
\item For every fixed number $y$ there exists a small interval $A$ such that $\varphi(\mcZ) = y$ implies $\Phi(\mcZ) \in A$. In other words, the set $\Phi \left( \varphi^{-1} \left( \{y\} \right) \right)$ is small.
\end{enumerate}
Assuming that $\varphi$ satisfies the first criterion, we will define its equitability on $\Q$ with respect to $\Phi$ as the extent to which it satisfies the second. A stronger version of this definition, which we will call {\em perfect equitability}, adds the requirement that the interval $A$ be of size $0$, i.e. that $\Phi(\mcZ)$ be exactly recoverable from $\varphi(\mcZ)$ regardless of the particular identity of $\mcZ \in \Q$. Note, however, that this is strictly a special case of our definition, and in general when we discuss equitability we are explicitly acknowledging the fact that this may not be the case.

The notion of equitability we just described for $\varphi$ will then have a natural extension to a statistic $\hvphi$. The extension will proceed in the same way that one might define a confidence interval of an estimator: for a fixed number $y$, instead of considering the distributions $\mcZ$ for which $\varphi(\mcZ) = y$ exactly, we will consider the distributions $\mcZ$ for which $y$ is a likely result of computing $\hvphi(z_1, \ldots, z_n)$, where $z_1, z_2, ..., z_n \overset{iid}{\sim} \mcZ$.

In Section~\ref{sec:equitAndPower} we will use this formalization of equitability to give an alternate, equivalent definition in terms of statistical power against a certain set of null hypotheses.

Though \cite{MINE} focused primarily on various models of noisy functional relationships with $R^2$ as the property of interest, the appropriate definitions of $\Q$ and $\Phi$ may change from application to application. For instance, as we have noted previously, when $\Q$ is the set of all superpositions of {\em noiseless} functional relationships and $\Phi \equiv 1$, then the population $\MIC$ (i.e. $\popMIC$) is perfectly equitable. More generally, instead of functional relationships one may be interested in relationships supported on one-manifolds, with added noise. Or perhaps instead of $R^2$ one may decide to focus simply on the magnitude of the added noise, or on the mutual information between the sampled y-values and the corresponding de-noised y-values. In each case the overarching goal should be to have $\Q$ be as large as possible without making it impossible to define an interesting $\Phi$ or making it impossible to find a measure of dependence that achieves good equitability on $\Q$ with respect to this $\Phi$. For this reason, we keep our exposition on equitability generic, and use noisy functional relationships and $R^2$ only as as a motivating example.

Keeping our exposition generic also allows us to address variations on the concept of equitability that have been introduced by others. For example, we are able to state in a formal, unified language the relationship of the work of Kinney and Atwal \cite{kinney2014equitability} to our previous work on $\MIC$. In particular, we explain why their negative result about the impossibility of achieving perfect equitability is of limited scope due to its focus on perfect equitability and to the setting of $\Q$ that it requires. (See Section~\ref{subsubsec:kinney} for this discussion.)

As a matter of record, we wish to clarify at this point that the key motivation given for Kinney and Atwal's work, namely that our original paper \cite{MINE} stated that $\MIC$ was perfectly equitable, is incorrect. Specifically, they write ``The key claim made by Reshef et al. in arguing for the use of MIC as a dependence measure has two parts. First, MIC is said to satisfy not just the heuristic notion of equitability, but also the mathematical criterion of $R^2$-equitability...", with the latter term referring to perfect equitability~\cite{kinney2014equitability}. However, such a claim was never made in \cite{MINE}. Rather, that paper informally defined equitability as an approximate notion and compared the equitability of $\MIC$, mutual information estimation, and other schemes empirically, concluding that $\MIC$ is the most equitable statistic in a variety of settings. In other words, one method can be more equitable than another, even if neither method is perfectly equitable. We intend for the formal definitions we present in this section to lead to a clearer picture of the relationships among these concepts and among the results published about them.

\subsection{Preliminaries: interpretability and reliability}
Let $\P$ be the set of distributions over $\R^2$, and let $\varphi : \P \rightarrow [0,1]$ be a mapping such that for $\mcZ \in \P$ describing a pair of jointly distributed random variables, $\varphi(\mcZ) = 0$ if and only if $X$ and $Y$ are statistically independent. Such a map $\varphi$ is called a {\em measure of dependence}.

Now let $\Q = \{\mcZ_\theta : \theta \in \Theta \} \subset \P$ be some subset of $\P$ indexed by a parameter $\theta \in \Theta$, and let $\Phi : \Q \rightarrow [0,1]$ be some property that is defined on $\Q$ but may not be defined on all of $\P$. We ask the following question: to what extent does knowing $\varphi(\mcZ)$ for some $\mcZ \in \Q$ tell us about the value of $\Phi(\mcZ)$? We will refer to the members of $\Q$ as {\em standard relationships}, and to $\Phi$ as the {\em property of interest}.

Conventionally, noisy functional relationships have been used as standard relationships, and the corresponding property of interest has been $R^2$ with respect to the regression function. However, as noted above we might imagine different scenarios. For this reason, we will make our exposition as generic as possible and refer back to the setting of noisy functional relationships as a motivating example.

\begin{figure}[h]
	\centering
	\begin{tabular}[c]{cc}
        \includegraphics[clip=true, trim = 0in 0in 0in 0in, width=0.4\textwidth]{\pathToCommonFigs/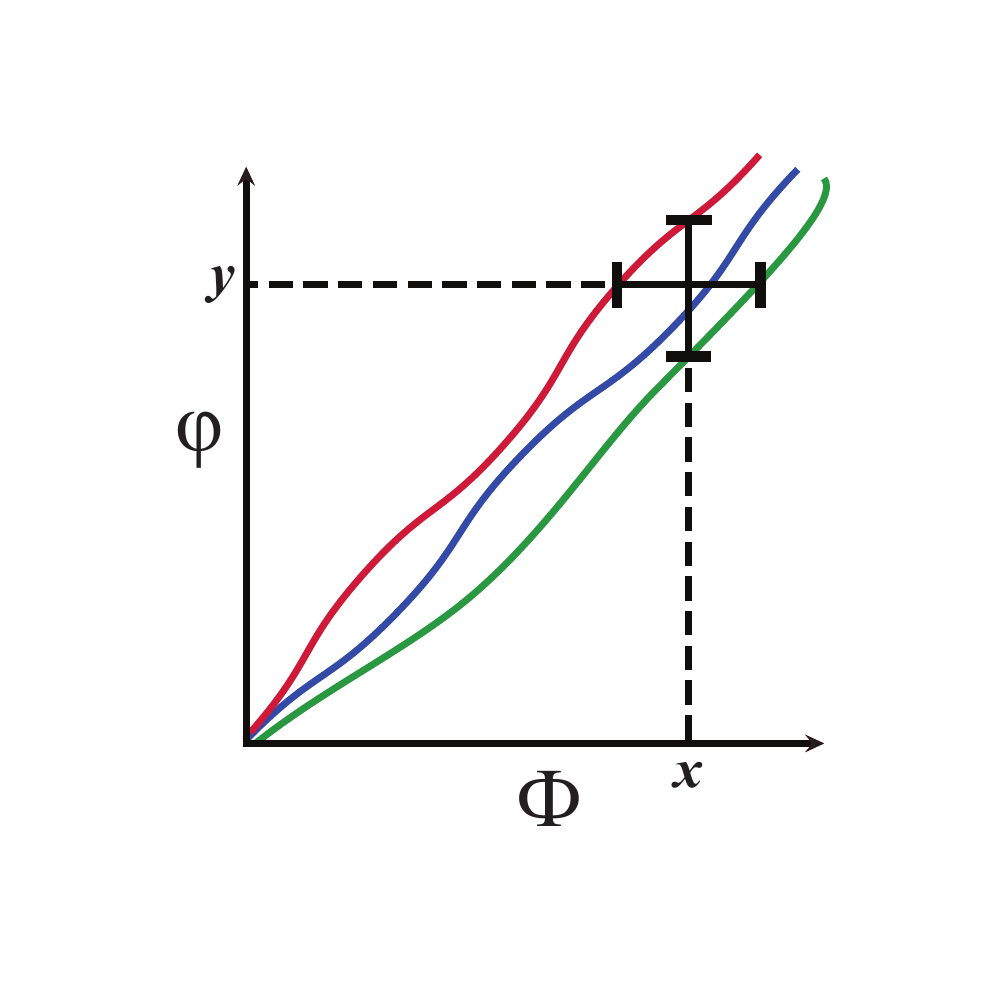}
                  &
        \includegraphics[clip=true, trim = 0in 0in 0in 0in, width=0.4\textwidth]{\pathToCommonFigs/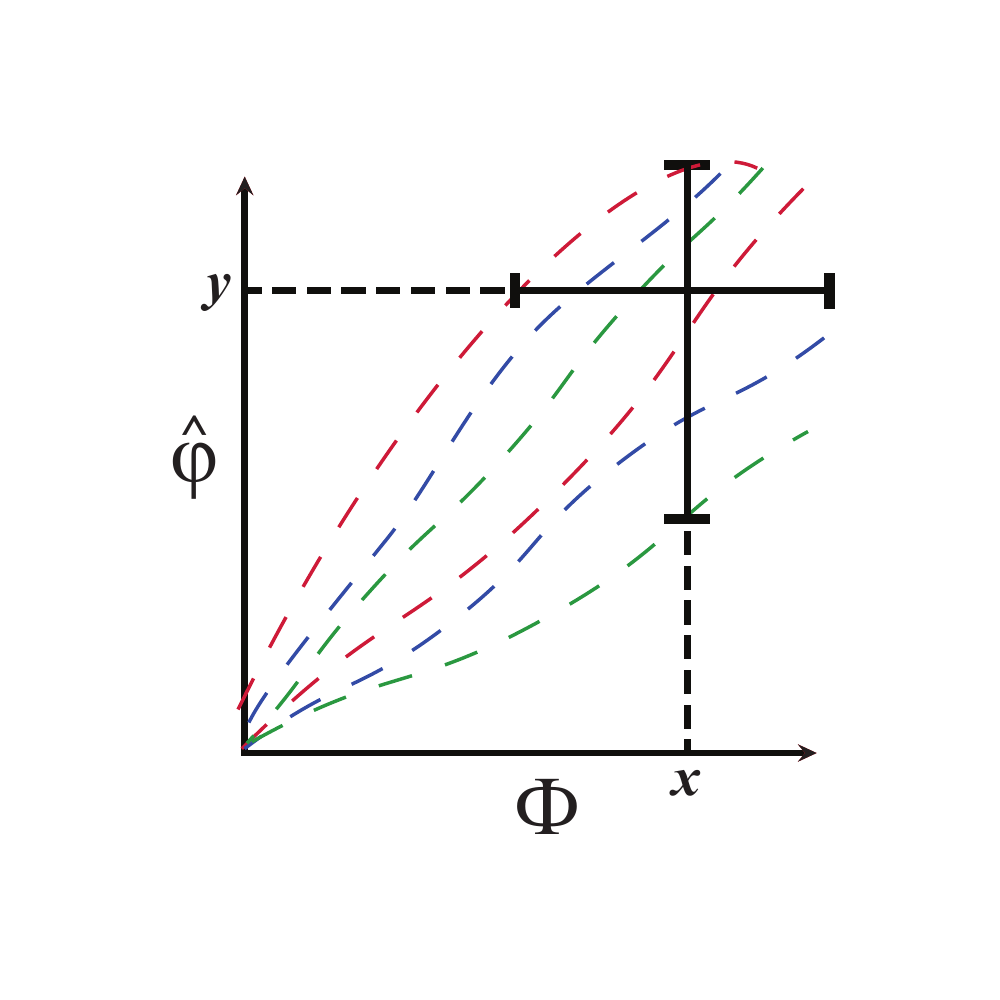} \\
       (a) & (b) \\
   \end{tabular}
    
  \caption{An schematic illustration of reliable and interpretable intervals. In both figure parts, $\Q$ is a union of three different models corresponding to the three different colors. (a) The relationship between $\varphi$ and $\Phi$ on distributions in $\Q$ in the infinite-data limit. The indicated vertical interval is the reliable interval $\reliable{\varphi}{x}$, and the indicated horizontal interval is the interpretable interval $\interpretable{\varphi}{y}$.  (b) The relationship between some estimator $\hvphi$ of $\varphi$ and $\Phi$ on $\Q$ at a finite sample size. The colored dashed lines indicate the $\alpha/2$ and $1-\alpha/2$ percentiles of the sampling distribution of $\hvphi$ for each model, at various values of $\Phi$. The indicated vertical interval is the reliable interval $\reliablestat{\alpha}{\hvphi}{x}$, and the indicated horizontal interval is the interpretable interval $\interpretablestat{\alpha}{\hvphi}{y}$.} \label{fig:reliabilityInterpretability}
\end{figure}

Regardless of our choice of $\varphi$, $\Q$, and $\Phi$, there are two straightforward ways to measure how similar $\varphi$ is to $\Phi$ on $\Q$. The first such way is to restrict our attention only to distributions $\mcZ$ with $\Phi(\mcZ) = x$, and then to ask how much $\varphi(\mcZ)$ can vary subject to that constraint.

\begin{definition}
Let $\varphi : \P \rightarrow [0,1]$ be a measure of dependence, and let $x \in [0,1]$. The smallest closed interval containing the set $\varphi(\Phi^{-1}(\{x\}))$ is called the {\em reliable interval} of $\varphi$ at $x$ and is denoted by $\reliable{\varphi}{x}$. $\varphi$ is a {\em $\gamma$-reliable} proxy for $\Phi$ on $\Q$ at $x$ if and only if the diameter of $\reliable{\varphi}{x}$ is at most $1/\gamma$.
\end{definition}
Equivalently, $\varphi$ is a $\gamma$-reliable proxy for $\Phi$ on $\Q$ at $x$ if and only if there exists an interval $A$ of size $1/\gamma$ such that $\Phi(\mcZ) = x$ implies that $\varphi(\mcZ) \in A$. In other words, if we restrict our attention to distributions $\mcZ$ such that $\Phi(\mcZ) = x$ we are guaranteed that $\varphi$ applied to those distributions will produce values that are close to each other. (See Figure~\ref{fig:reliabilityInterpretability}a for an illustration.) In the context of noisy functional relationships and $R^2$, this corresponds to saying that relationships with the same $R^2$ will not score too differently.

The second way of measuring how closely $\varphi$ matches $\Phi$ on $\Q$ is to talk about how much $\Phi(\mcZ)$ can vary when we consider only distributions $\mcZ$ with $\varphi(\mcZ) = y$.
\begin{definition}
Let $\varphi : \P \rightarrow [0,1]$ be a measure of dependence, and let $y \in [0,1]$. The smallest closed interval containing the set $\Phi( \varphi^{-1}( \{ y \}))$ is called the {\em interpretable interval} of $\varphi$ at $y$ and is denoted by $\interpretable{\varphi}{y}$. $\varphi$ is a {\em $\gamma$-interpretable} proxy for $\Phi$ on $\Q$ at $y$ if and only if the diameter of $\interpretable{\varphi}{y}$ is at most $1/\gamma$.
\end{definition}
Equivalently, $\varphi$ is a $\gamma$-interpretable proxy for $\Phi$ on $\Q$ at $y$ if and only if there exists an interval $A$ of size $1/\gamma$ such that $\varphi(\mcZ) = y$ implies that $\Phi(\mcZ) \in A$ for all $\mcZ \in \Q$. In other words, if all we know about a distribution $\mcZ$ is that $\varphi(\mcZ) = y$, then we are able to guess what $\Phi(\mcZ)$ is pretty accurately. (See Figure~\ref{fig:reliabilityInterpretability}a for an illustration.) In the context of noisy functional relationships and $R^2$, this corresponds to the fact that evaluating $\varphi$ on a relationship will give us good upper and lower bounds on the noise-level of that relationship as measured by $R^2$.

When $\Phi$ and $\Q$ are clear we will omit them and describe $\varphi$ simply as $\gamma$-reliable (resp. interpretable) at $x$ (resp. $y$).

Once we have specified what we mean by ``reliable" and ``interpretable", it is straightforward to define ``reliability" and ``interpretability".

\begin{definition}
The {\em reliability} (resp. {\em interpretability}) of $\varphi$ at $x$ (resp. $y$) is $1/d$, where $d$ is the diameter of $\reliable{\varphi}{x}$ (resp. $\interpretable{\varphi}{y}$). If $d=0$, the reliability (resp. interpretability) of $\varphi$ is $\infty$ and $\varphi$ is called {\em perfectly reliable} (resp. {\em interpretable}).
\end{definition}
We will occasionally refer to the more general notions of reliability/interpretability as ``approximate" to distinguish them from the perfect case.

One can imagine many different ways to quantify the overall interpretability and reliability of a measure of dependence. For instance, we have
\begin{definition}
A measure of dependence is {\em worst-case} $\gamma$-reliable (resp. interpretable) if it is $\gamma$-reliable (resp. interpretable) at all $x$ (resp. $y$) $\in [0,1]$.

A measure of dependence is {\em average-case} $\gamma$-reliable (resp. interpretable) if its reliability (resp. interpretability), averaged over all $x$ (resp. $y$) $\in [0,1]$, is at least $\gamma$.
\end{definition}

More generally, one could imagine defining a prior over all the distributions in $\Q$ to reflect one's belief about the importance of various types of relationships in the world, and then using that to measure overall reliability and interpretability. We do not pursue this here; instead, we focus only on worst-case reliability and interpretability.

Let us give two simple examples of the use of this new terminology. First, the Linfoot correlation coefficient~\cite{linfoot1957informational}, defined as $1 - 2^{-2I}$ where $I$ is mutual information, is a worst-case perfectly interpretable and perfectly reliable proxy for $\rho^2$, the squared Pearson correlation coefficient $\rho^2$, on the set $\Q$ of bivariate normal random variables. Additionally, Theorem 6 of~\cite{szekely2009brownian} implies that distance correlation is a perfectly interpretable and perfectly reliable proxy for $| \rho |$ on the same set $\Q$. In the first example, the given measure of dependence simply equals $\rho^2$ when it is restricted to $\Q$, which is why the reliability and interpretability are perfect. In the second example, the distance correlation does not equal $|\rho|$ but rather is a deterministic function of it, which is sufficient.

\subsection{Defining equitability}

\subsubsection{Equitability in the sense of \cite{MINE}}
As we have suggested above, in the language of reliability and interpretability, the informal notion of equitability described in \cite{MINE} amounts to a requirement that a measure of dependence be a highly interpretable proxy for some property of interest $\Phi$ that is suitably defined to reflect ``strength'', and over as large a model $\Q$ as possible.

We have discussed the fact that the particular choice of $\Phi$ and $\Q$ may vary from problem to problem, as might the way in which the equitability is measured (average-case versus worst-case). Let us define the models $\Q$ considered in \cite{MINE}. We begin by stating precisely what we mean by the term ``noisy functional relationship".

\begin{definition}
A random variable distributed over $\R^2$ is called a {\em noisy functional relationship} if and only if it can be written in the form $(X + \ep, f(X) + \ep')$ where $f : [0,1] \rightarrow \R$, $X$ is a random variable distributed over $[0,1]$, and $\ep$ and $\ep'$ are random variables. We denote the set of all noisy functional relationships by $\F$.
\end{definition}
As we will soon discuss, there are varying views about whether constraints should be placed on $\ep$ and $\ep'$, ranging from setting them to be Gaussians independent of each other and of $X$ all the way to allowing them to be arbitrary random variables that are not necessarily independent of $X$. For this reason, we do not place any constraints on them in the above definition.

With the concept of noisy functional relationships defined, equitability on a set of functional relationships simply amounts to the use of $R^2$ as the property of interest.
\begin{definition}[Equitability on functional relationships in the sense of Reshef et al.]
Let $\Q \subset \F$ be a set of noisy functional relationships. A measure of dependence is worst-case (resp. average-case) {\em $\gamma$-equitable} on $\Q$ if it is a worst-case (resp. average case) $\gamma$-interpretable proxy for $R^2$ on $\Q$.
\end{definition}
In this paper, we will often use ``equitability" with no qualifier to mean worst-case equitability.

Given a set $F$ of functions from $[0,1]$ to $\R$, \cite{MINE} defined a few different subsets of $\F$. The simplest is
\[ \Q_F^{Y,U} = \{ (X, f(X) + \ep_b) : b \in \R_{\geq 0}, f \in F \} \]
Where the letter U in $\Q_F^{Y,U}$ indicates that $X$ is uniform\footnote{In \cite{MINE} $X$ was not actually random. Instead, values of $X$ were chosen in $[0,1]$ to produce $n$ evenly spaced x-values. However, for theoretical clarity we opt here to treat $X$ as a random variable.} over $[0,1]$, and $\ep_b$ is uniform over $[-b,b]$ and is independent of $X$. Of course, one can add noise in the first coordinate as well, producing
\[ \Q_F^{XY,U} = \{ (X + \ep_a, f(X) + \ep_b) : a,b \in \R_{\geq 0}, f \in F \} \]
where $\ep_a$ is defined analogously to $\ep_b$. In both of the above cases, we can also modify $X$ such that, rather than being uniformly distributed over $[0,1]$, it is distributed in such a way that $(X, f(X))$ is uniformly distributed over the graph of $f$. This gives the last two models, $\Q_F^{Y,G}$ and $\Q_F^{XY,G}$, that are used in \cite{MINE}.

The reason that \cite{MINE} defined four different models models was simple: since it is often difficult to say exactly which model (if any) is actually followed by real data, we would ideally like to see good equitability on as many different such models as possible. Given the lack of a neat description of how real data behave, we aim for robustness.

Nevertheless, each of these models is somewhat narrow, and we can easily imagine others: for instance, we might define $\ep_a$ and $\ep_b$ to be Gaussian, we might allow them to depend on each other, or we might consider adding noise only to the first coordinate. Each of these modifications deserves attention.

\begin{rmk}
In the remainder of this paper, we will use the terms ``equitability" and ``interpretability" differently, but the difference is merely notional and not formal: equitability is a type of interpretability that we get when our goal is that $\Phi$ reflect the strength of our relationships.
\end{rmk}

\subsubsection{Kinney and Atwal's impossibility result}
\label{subsubsec:kinney}

Now that we have a sufficiently general language in which to discuss equitability, let us turn to the recent impossibility result of Kinney and Atwal~\cite{kinney2014equitability}. Kinney and Atwal write the following.
\begin{quotation}
[W]e prove that the definition of equitability proposed by Reshef et al. (ed: \cite{MINE}) is, in fact, impossible for any (nontrivial) dependence measure to satisfy.
\end{quotation}

However, this result actually has two severe limitations to its scope. To understand these issues, let us state the result in the language developed above: it amounts to showing that no non-trivial measure of dependence can be perfectly equitable (i.e. a perfect worse-case interpretable proxy for $R^2$) on $\Q_K$, where
\[ \Q_K = \left\{ (X, f(X) + \eta) \mbox{ } \big| \mbox{ } f : [0,1] \rightarrow [0,1], (\eta \perp X) | f(X)\right\} \]
with $\eta$ representing a random variable that is conditionally independent of $X$ given $f(X)$. This model describes functional relationships with noise in the second coordinate only, where that noise can depend arbitrarily on the value of $f(X)$ (i.e. it can be heteroscedastic) but must be otherwise independent of $X$.

The first limitation of this result is that the argument depends crucially on the fact that the noise term $\eta$ can depend arbitrarily on the value of $f(X)$. In particular, its mean need not be 0 but rather may change depending on $f(X)$. As pointed out in~\cite{Murrell2014comment}, selecting such a large model $\Q_K$ leads to identifiability issues such as allowing one to obtain the relationship $f(X) = X^2$ as a noisy version of $f(X) = X$. The more permissive (i.e. large) a model is, the easier it is to prove an impossibility result for it, and $\Q_K$ is indeed quite large: in particular, it is not contained in any of the models $\Q_F$ defined above. This would be necessary in order for impossibility on $Q_K$ to translate into impossibility for one of these other models. Thus, Kinney and Atwal's result does not apply to the models $\Q_F$ defined in \cite{MINE}.

The second limitation of Kinney and Atwal's result is that it only addresses {\em perfect} equitability rather than the more general, approximate notion with which we are primarily concerned. As we discussed in Section~\ref{subsec:equitOverview}, the claim that the definition of equitability given in \cite{MINE} was one of perfect equitability rather than approximate equitability is incorrect. More generally however, though a perfectly equitable proxy for $R^2$ may indeed be difficult or even impossible to achieve for many large models $\Q$ including some of the models $\Q_F$ defined above, such impossibility would make {\em approximate} equitability no less desirable a property. The question thus remains how equitable various measures are, both provably and empirically. To borrow an analogy from computer science, the fact that a problem is proven to be NP-complete does not mean that we that we do not want efficient algorithms for the problem; we simply may have to settle for heuristic solutions, or solutions with some provable approximation guarantees. Similarly, there is merit in searching for measures of dependence that appear to be sufficiently equitable proxies for $R^2$ in practice.

For more on this discussion, see the technical comment~\cite{reshef2014comment} published by the authors of this paper about~\cite{kinney2014equitability}.

\subsection{Equitability of a statistic}
Until now we have only discussed the properties of a measure of dependence $\varphi$ considered as a function of random variables. However, it is trivial to define a perfectly reliable and interpretable proxy for any $\Phi$ on any $\Q$: simply define $\varphi$ to equal $\Phi$ on $\Q$ and an arbitrary measure of dependence on $\P \backslash \Q$. Of course, this is not the point. Rather, the idea is to define a function $\varphi$ that is amenable to efficient estimation, and to use the notions of interpretability and reliability defined above in order to separate the loss in performance that a given estimator of $\varphi$ incurs from finite sample effects from the loss in performance caused by the choice of the estimand $\varphi$ itself.

\begin{figure}[h!]
	\centering
    \includegraphics[clip=true, trim = 0in 4in 0in 0in, height=0.4\textheight]{\pathToCommonFigs/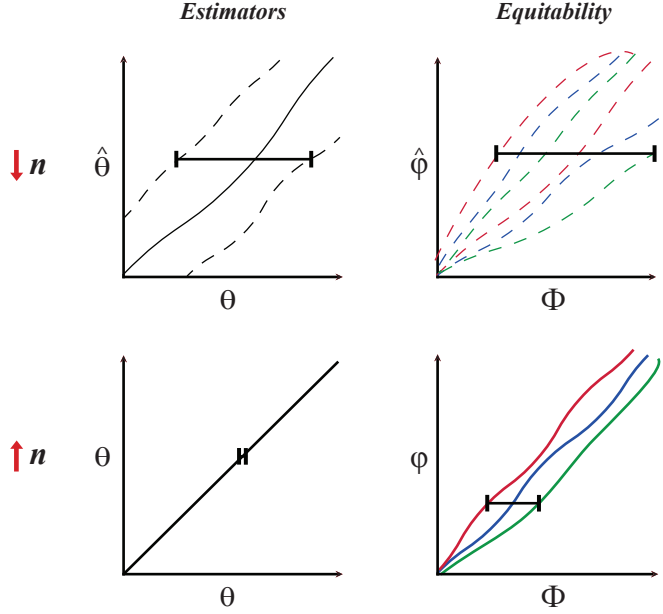}
  \caption{The analogy between interpretable intervals and confidence intervals. The left-hand column depicts a scenario in which $\hvphi$ is estimating a parameter $\theta$. As sample size increases, the width of the confidence intervals of $\hvphi$ will tend to zero because each value of $\theta$ corresponds to exactly one population value of $\hvphi$. The right-hand column depicts a scenario in which $\hvphi$ is being used as an estimate of $\Phi$, but $\Phi$ does not completely determine the population value of $\hvphi$: the red, blue, and green curves represent distinct sets of distributions in $\Q$ whose members can have identical values of $\Phi$. For instance, they might correspond to different function types. This is the setting in which we are operating, and the intervals plotted on the right are called interpretable intervals. Interpretable intervals can be large either because of finite sample effects (as in the conventional estimation case) or because of the lack of interpretability of the population value of the statistic (shown in the bottom-right picture).} \label{fig:analogy}
\end{figure}

However, to reason about this distinction, we do need a way to directly evaluate the reliability and interpretability of a statistic at a given sample size. To do so, we will adapt our above definitions from the ``infinite-data limit" by analogy using the theory of estimation and confidence intervals. Specifically, in estimation theory, confidence intervals can be defined in terms of the sets of likely values of a statistic at each value of the parameter. In the same way, we will define a {\em reliable interval} to be a set of likely values of $\hvphi$ given a certain value of $\Phi$, and then define the {\em interpretable interval} in terms of the values of $\Phi$ whose reliable intervals contain a given value of $\hvphi$. This analogy is depicted in Figure~\ref{fig:analogy} and Table~\ref{tab:analogy}.

\begin{rmk}
The analogy between an equitable statistic and an estimator with small confidence intervals can be made even more explicit as follows: ordinarily, the best way to obtain information about $\Phi$ would be to estimate it directly. However, if we do so we are not guaranteed that the statistic we use will detect any deviation from statistical independence when used on distributions not in $\Q$. Thus, our problem is akin to that of seeking the best possible estimator $\hvphi$ of $\Phi$ on $\Q$ subject to the constraint that the population value $\varphi$ equal 0 if and only if the distribution in question exhibits statistical independence. The difference is that we only care about the confidence intervals of the estimator and not about its bias, since we are principally interested in ranking relationships according to $\Phi$ rather than recovering the exact value of $\Phi$.
\end{rmk}

\begin{table}[h]
\centering
	\begin{tabular}{p{0.5in}p{2in}p{3.0in}}
	&
	\textbf{Estimating $\theta$ (confidence)} &
	\textbf{``Estimating" $\Phi$ (interpretability)} \\ \hline \\
	Model &
	One value of $\varphi$ for each value of $\theta$ &
	Multiple values of $\varphi$ for each value of the $\Phi$ \vspace{1\bigskipamount} \\ \hline \\
	 
	Error &
	Confidence intervals wide due to finite sample effects &
	Interpretable intervals wide due to finite sample effects, as well as infinite sample relationship between $\varphi$ and $\Phi$ \vspace{1\bigskipamount} \\ \hline \\
	
	Tests &
	Small confidence intervals give power at rejecting $H_0 : \theta < \theta_0$   &
	Small interpretable intervals give power at rejecting $H_0 : \Phi < \Phi_0$ \vspace{1\bigskipamount}
\end{tabular}
\caption{The analogy between confidence intervals in the setting of estimating a parameter $\theta$ that completely parametrizes a model, and interpretable intervals when viewed as confidence intervals for $\hvphi$ as an ``estimate" of $\Phi$.} \label{tab:analogy}
\end{table}

We first define the reliability of a statistic. Previously, reliability meant that if we know $\Phi(\mcZ)$ then we can place $\varphi(\mcZ)$ in a small interval. To obtain the analogous definition for a statistic, we simply relax the requirement that $\varphi(\mcZ)$ be in a small interval to the requirement that $\hvphi(D)$ be in a small interval with high probability when $D$ is a sample from $\mcZ$. This is equivalent to simply considering $\hvphi$ as an estimator of $\Phi$ rather than of $\varphi$ and requiring that its sampling distribution have its probability mass concentrated in a small area.

\begin{definition}
Let $\hvphi : \R^{2n} \rightarrow [0,1]$ be a statistic, let $x, \alpha \in [0,1]$. The $\alpha$-reliable interval\footnote{This is simply the union of the central intervals of the sampling distribution of $\hvphi$ taken over all distributions $\mcZ \in \Phi^{-1}(\{x\})$. } of $\hvphi$ at $x$, denoted by $\reliablestat{\alpha}{\hvphi}{x}$, is the smallest closed interval $A$ with the property that, for all $\mcZ \in \Q$ with $\Phi(\mcZ) = x$,
\[ \Pr{\hvphi(D) < \min A} < \alpha \]
and
\[ \Pr{\hvphi(D) > \max A} < \alpha \]
where $D$ is a sample of size $n$ from $\mcZ$.

The statistic $\hvphi$ is a {\em $\gamma$-reliable proxy} for $\Phi$ on $\Q$ at $x$ with probability $1-2\alpha$ if and only if the diameter of $\reliablestat{\alpha}{\hvphi}{x}$ is at most $1/\gamma$.
\end{definition}
(See Figure~\ref{fig:reliabilityInterpretability}b for an illustration.) Looking once more at the example of noisy functional relationships with $R^2$ as the property of interest, this corresponds to the requirement that there exist an interval $A$ such that, for any functional relationship $\mcZ$ with an $R^2$ of $x$, $\hvphi(D)$ falls within $A$ with high probability when $D$ is a sample from $\mcZ$.

Once reliability is suitably defined, the definition of interpretability is simple to translate into one for a statistic. Here we again make our definition by considering $\hvphi$ as an estimator of $\Phi$ and looking at its confidence intervals. The key is that while we generally think of a confidence interval of a consistent estimator becomes large only due to finite sample effects, the so-called interpretable interval can become large either because of finite sample effects or because the function $\varphi$ to which $\hvphi$ converges is itself not very interpretable.
\begin{definition}
Let $\hvphi : \R^{2n} \rightarrow [0,1]$ be a statistic, and let $y, \alpha \in [0,1]$. The $\alpha$-interpretable interval of $\hvphi$ at $y$, denoted by $\interpretablestat{\alpha}{\hvphi}{y}$, is the smallest closed interval containing the set
\[ \left\{ x \in [0,1] : y \in \reliablestat{\alpha}{\hvphi}{x} \right\} \]

The statistic $\hvphi$ is a {\em $\gamma$-interpretable proxy} for $\Phi$ on $\Q$ at $y$ with confidence $1-2\alpha$ if and only if the diameter of $\interpretablestat{\alpha}{\hvphi}{y}$ is at most $1/\gamma$.
\end{definition}
(See Figure~\ref{fig:reliabilityInterpretability}b for an illustration.)

\begin{rmk}
Note that our definitions do not require that $\hvphi$ converge to $\Phi$ in any sense; we are not trying to construct a measure of dependence that also estimates $\Phi$ exactly. Rather, we are willing to tolerate some discrepancy between $\hvphi$ and $\Phi$ in order to preserve the fact that $\hvphi$ acts as a measure of dependence when applied to samples from distributions not in $\Q$. This is the essential compromise behind the idea of equitability. Why is it worthwhile to make? Because on the one hand, if we are interested in ranking relationships then having only a measure of dependence with no guarantees about how noise affects its score will not do; but on the other hand, we want a statistic that is robust enough that we will not completely miss relationships that do not fall in this set.
\end{rmk}

Analogous definitions can be made for average-case and worst-case reliability/equitability, and for equitability on functional relationships.

\subsection{Discussion}
As the definitions given above imply, an equitable statistic is different from other measures of dependence in that its main intended use is not testing for independence, but rather measurement of effect size. The idea is to have a statistic that has the robustness of a measure of dependence but that also, via its relationship to $\Phi$, gives values that have a clear, if approximate, interpretation and can therefore be used to rank relationships.

There is a tension inherent in the concept of equitability that arises from the attempt to reconcile the robustness of a measure of dependence with the utility of a measure of effect size. This tension leads to two important concessions to pragmatism.
\begin{enumerate}
\item The set $\Q$ is not the set $\P$ of all distributions but rather some strict subset of it.
\item Despite the fact that we evaluate $\hvphi$ as an estimator of $\Phi$, we have not required that $\hvphi$ converge to $\Phi$ in any sense, and we explicitly allow for the possibility that it may not. Rather, we are willing to tolerate some discrepancy between the population value of $\hvphi$ and $\Phi$ in order to preserve the fact that $\hvphi$ acts as a measure of dependence when applied to samples from distributions not in $\Q$.
\end{enumerate}
The first of these compromises necessitates the second. For if we could set $\Q$ to be the set of all distributions and still define a property of interest $\Phi$ that captured what we mean by a ``strong" relationship, then we truly would simply seek an estimator for $\Phi$ and be done. Unfortunately we cannot do this; the concepts of ``noise" and what it means to be a ``strong" relationship can become elusive when we enlarge $\Q$ too much. However, this does not mean that we should give up on seeking a statistic that somehow performs reasonably at ranking relationships. Therefore, while define exactly what we would like to have (i.e, $\Phi$) whenever we can (i.e., on some $\Q \subsetneq \P$), we still demand that our statistic act as a measure of dependence on relationships not in $\Q$. This second requirement may hurt our ability to estimate $\Q$, but when we are exploring data sets with real relationships whose form we cannot fully anticipate or model, the robustness it gives can be worth the price of relaxing the requirement that $\hvphi$ converge to $\Phi$ to a requirement that it merely approximate $\Phi$. This is our second compromise.

In this section, we largely focused on setting $\Q$ to be some subset of the set $\F$ of noisy functional relationships, as this has been the subject of most of the empirical work on the equitability of $\MIC$ and other measures of dependence. However, it is important to keep in mind that $\Q$ should ideally be larger than this. For instance, as we discussed previously, in \cite{MINE} the equitability of $\MIC$ is discussed not just in the case of noisy functional relationships but also in the case of superpositions of functional relationships.

As the compromises discussed above make clear, equitability sits in between the traditional hypothesis-testing paradigm of measures of dependence on the one hand and the paradigm of measuring effect size on the other. However, equitability can actually be framed entirely in terms of hypothesis tests. This is the topic of our next section.

\section{Equitability as a generalization of power against independence}
\label{sec:equitAndPower}
Having defined equitability in terms of estimation theory, we will now show that we can equivalently think of it in terms of power against a certain family of null hypotheses. This result re-casts equitability as a generalization of power against statistical independence and gives a second formal definition of equitability that is easily quantifiable using traditional power analysis.

\subsection{Overview}
Our proof is based on the idea behind a standard construction of confidence intervals via inversion of statistical tests. In particular, equitability of a statistic $\hvphi$ with respect to a property of interest $\Phi$ on a model $\Q$ will be shown to be equivalent to power against the collection of null hypotheses of the form $\{ H_0^\alpha : \Phi(\mcZ) \leq \alpha, \mcZ \in \Q \}$ corresponding to different values $\alpha$ of $\Phi$. Thus, if $\Phi$ is such that $\Phi(\mcZ) = 0$ if and only if $\mcZ$ exhibits statistical independence, then equitability with respect to $\Phi$ is a strictly stronger requirement than power against statistical independence.

As a concrete example, let us again return to the case in which $\Q$ is a set of noisy functional relationships and the property of interest is $R^2$. Here, a conventional power analysis would consider, say, the right-tailed test based on the statistic $\hvphi$ and evaluate its type II error at rejecting the null hypothesis of $R^2 = 0$, i.e. statistical independence. In contrast, we will show that for $\hvphi$ to be equitable, it must yield right-tailed tests with high power against null hypotheses of the form $R^2 \leq a$ for {\em any} $a \geq 0$. This is difficult: each of these new null hypotheses can be composite since $\Q$ can contain relationships of many different types (e.g. a noisy linear relationship, a noisy sinusoidal relationship, and a noisy parabolic relationship). Whereas all of these relationships may have reduced to a single null hypothesis of statistical independence in the case of $R^2 = 0$, they yield composite null hypotheses once we allow $R^2$ to be non-zero.

\subsection{Definitions and proof of the result}
As before, let $\P$ be the set of distributions over $\R^2$, and let $\varphi : \P \rightarrow [0,1]$ be a measure of dependence estimated by some statistic $\hvphi : \R^{2n} \rightarrow [0,1]$. Let $\Q \subset \P$ be some model of interest and let $\Phi : \Q \rightarrow [0,1]$ be a property of interest.

Now, given some $x_0, x, c \in [0,1]$, let $T^{x_0}(x, c)$ be the right-tailed test based on $\hvphi$ with critical value $c$, null hypothesis $H_0 : \Phi(\mcZ) = x_0$, and alternative hypothesis $H_1 : \Phi(\mcZ) = x$. The set $S^{x_0}(x) = \{ T^{x_0}(x, c) : c \in [0,1] \}$ is the set of possible right-tailed tests based on $\hvphi$ that are available to us for distinguishing $H_1$ from $H_0$. We will distinguish a one of these tests in particular, namely the optimal one subject to a constraint on type I error: let $T_\alpha^{x_0}(x)$ be the test $T^{x_0}(x,c) \in S^{x_0}(x)$ with $c$ chosen to be as small as possible subject to the constraint that the type I error of the resulting test be at most $\alpha$. We are now ready to define the measure of power that we will use to show the equivalence with equitability.

\begin{definition}
Fix $\alpha, x_0 \in [0,1]$. For any given $x \in [0,1]$, let $\powerfunc{\alpha}{x_0}(x)$ be the power of $T_\alpha^{x_0}(x)$. We call the function $\powerfunc{\alpha}{x_0} : [0,1] \rightarrow [0,1]$ the {\em power function} associated to $\hvphi$ at $x_0$ with significance $\alpha$ with respect to $\Phi$.
\end{definition}

When $\Phi(\mcZ) = 0$ if and only if $\mcZ$ represents statistical independence, then the power function $\powerfunc{\alpha}{0}$ gives the power of right-tailed tests based on $\hvphi$ at distinguishing statistical independence from various non-zero values of $\Phi$ with significance $\alpha$. For instance, if $\Q$ is the set of bivariate normal distributions and $\Phi$ is ordinary correlation $\rho$, then $\powerfunc{\alpha}{0}(x)$ simply gives us the power of the right-tailed test based on $\hvphi$ at distinguishing the alternative hypothesis of $\rho = x$ from the null hypothesis of $\rho = 0$. As an additional example, in the cases discussed above where $\Q$ is some set of functional relationships and $\Phi$ is $R^2$, the power function $\powerfunc{\alpha}{0}(x)$ associated to $\hvphi$ equals the power of the right-tailed test based on $\hvphi$ that distinguishes the alternative hypothesis of $R^2 = x$ from the null hypothesis of $R^2 = 0$, i.e., independence, with type I error $\alpha$.

Nevertheless, as we observe here, the set of power functions at values of $x_0$ besides 0 contains much more information that just the power of right-tailed tests based on $\hvphi$ against the null hypothesis of $\Phi = 0$. We can recover the interpretability of a statistic at every $y \in [0,1]$ by considering its power functions at values of $x_0$ beyond 0. This is the main result of this section. It is analogous to the standard relationship between the size of the confidence intervals of an estimator and the power of their corresponding right-tailed tests.

\begin{rmk}
In this setup our null and alternative hypotheses, since they are based on $\Phi$ and not on a parametrization of $\Q$ that uniquely specifies distributions, may be composite: $\mcZ$ can be one of several distributions with $\Phi(\mcZ) = x_0$ or $\Phi(\mcZ) = x$ respectively. This composite nature of our null hypotheses is bound up in the reason we need interpretability and reliability in the first place: if the set $\Q$ were so small that each value of $\Phi$ defined only one distribution then we would likely not be in a setting where we needed an agnostic approach to detecting strong relationships. We could just estimate $\Phi$ directly.
\end{rmk}

Before we prove the main result of this section, the connection between power and interpretabiilty, we must first define what aspect of power will be reflected in the interpretability of $\hvphi$.
\begin{definition}
The {\em uncertain set} of a power function $\powerfunc{\alpha}{x_0}$ is the set $\{x \geq x_0 : \powerfunc{\alpha}{x_0}(x) < 1-\alpha \}$.
\end{definition}

We will now prove the main proposition of this section, which is essentially that uncertain sets are interpretable intervals and vice versa. In what follows, since our statistic $\hvphi$ is fixed we use $\reliablestat{\alpha}{}{x}$ to denote $\reliablestat{\alpha}{\hvphi}{x}$, and $\interpretablestat{\alpha}{}{y}$ to denote $\interpretablestat{\alpha}{\hvphi}{y}$. We also use the function $| \cdot |$ to denote the diameter of a subset of $[0,1]$.

\begin{prop} \label{prop:equitabilityAndPower}
Fix $0 < \alpha < 1/2$ and $d > 0$, and suppose $\hvphi$ is a statistic with the property that $\max \reliablestat{\alpha}{}{x}$ is a continuous, increasing function of $x$. The following two statements hold.

\begin{enumerate}
\item If $\left| \interpretablestat{\alpha}{}{y} \right| = d$, then the uncertain set of $\powerfunc{\alpha}{x_0}$ has diameter $d$ for $x_0 = \inf \{ x : y \in \reliablestat{\alpha}{}{x} \}$.

\item If the uncertain set of $\powerfunc{\alpha}{x_0}$ has diameter $d$, then $\left| \interpretablestat{\alpha}{}{y} \right| = d$ for $y = \max \reliablestat{\alpha}{}{x_0}$.
\end{enumerate}
\end{prop}

An illustration of this proposition and its proof is shown in Figure~\ref{fig:equitabilityAndPower}.
\begin{figure}[h!]
	\centering
    \includegraphics[clip=true, trim = 0in 0in 0in 0in, height=0.5\textheight]{\pathToCommonFigs/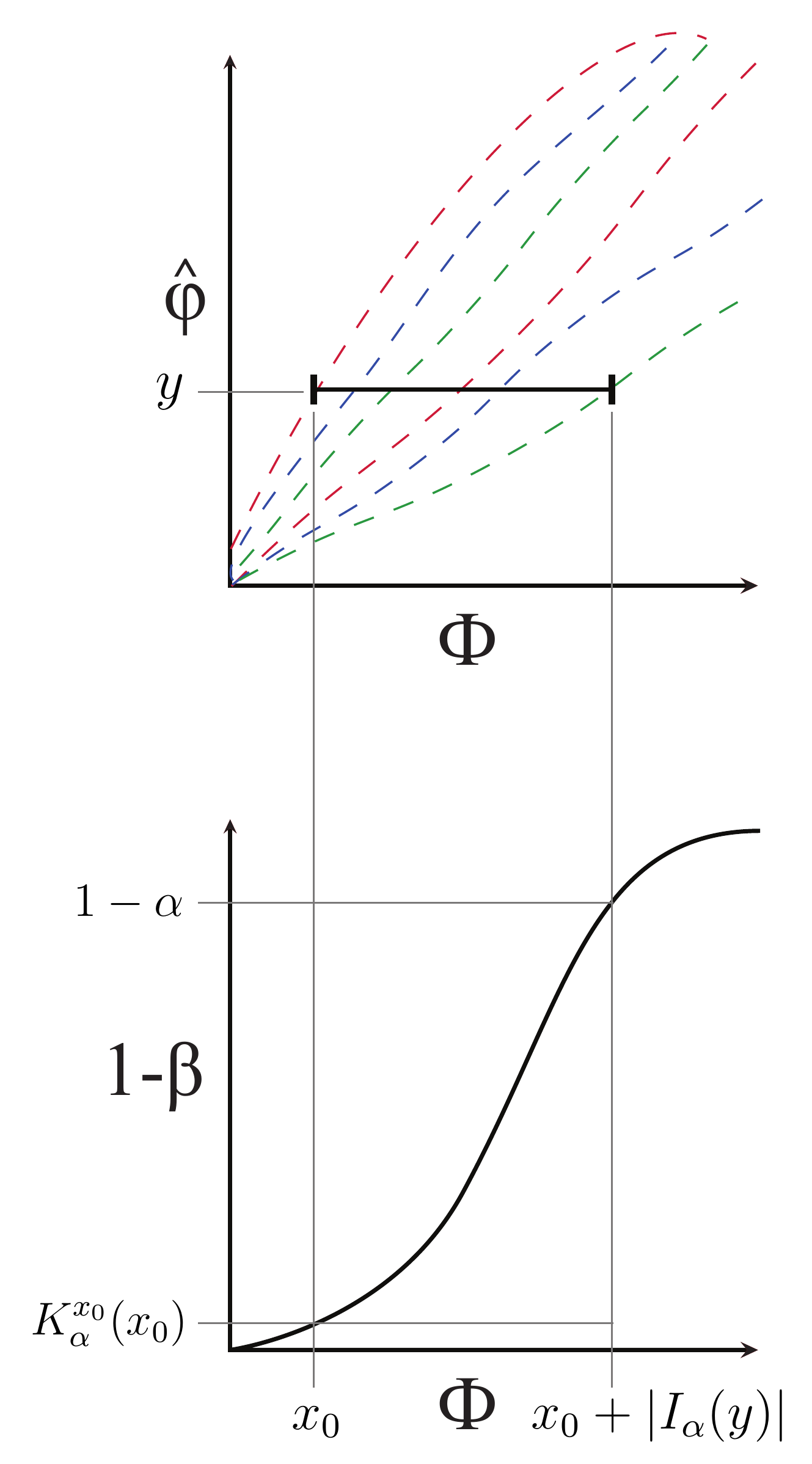}
  \caption{The relationship between equitability and power against independence, as in Proposition~\ref{prop:equitabilityAndPower}. The top plot is the same as the one in Figure~\ref{fig:reliabilityInterpretability}b, with the indicated interval denoting the interpretable interval $\interpretablestat{\alpha}{}{y}$. The bottom plot is a plot of the power function $\powerfunc{\alpha}{x_0}(x)$, with the y-axis indicating statistical power. The key to the proof of the proposition is to notice that the width of the interpretable interval describes the distance from $x_0$ to the point at which the power function reaches $1-\alpha$, and this is exactly the width of the uncertain set of the power function. (Notice that because the null and alternative hypotheses are composite, $\powerfunc{\alpha}{x_0}(x_0)$ need not equal $\alpha$; in general it may be lower.)}\label{fig:equitabilityAndPower}
\end{figure}

\begin{proof}
Let $T^{x_0}_\alpha(x)$ denote the statistical test corresponding to $\powerfunc{\alpha}{x_0}(x)$. We first determine: what is the critical value of $T^{x_0}_\alpha(x)$? By definition, it is the smallest critical value that gives a type I error of at most $\alpha$. In other words, it is the supremum, over all $\mcZ \in \Q$ with $\Phi(\mcZ) = x_0$, of the $1-\alpha$-percentile of the sampling distribution of $\hvphi$ when applied to $\mcZ$. But this is simply $\max \reliablestat{\alpha}{}{x_0}$.

We now prove the proposition by proving each of the two statements separately.

\proofwaypoint{Proof of the first statement} Let $U$ be the uncertain set of $\powerfunc{\alpha}{x_0}$. Since $\alpha < 1/2$, we know that $\powerfunc{\alpha}{x_0}(x_0) \leq \alpha < 1-\alpha$, and so $\inf U = x_0$. It therefore suffices to show that $\sup U = x_0 + d$.

We first show that $\sup U \geq x_0 + d$: since $\left| \interpretablestat{\alpha}{}{y} \right| = d$, we know that we can find $x$ arbitrarily close to $x_0 + d$ from below such that $y \in \reliablestat{\alpha}{}{x}$. But this means that there exists some $\mcZ \in \Q$ with $\Phi(\mcZ) = x$ such that if $D_x$ is a sample of size $n$ from $\mcZ$ then
\[ \Pr{\hvphi(D_x) < y} \geq \alpha \]
i.e.,
\[ \powerfunc{\alpha}{x_0}(x) = \Pr{\hvphi(D_x) \geq y} < 1- \alpha \]
and so $x \in U$.

We next show that $\sup U \leq x_0 + d$. To do so, we will need the following fact: since $\max \reliablestat{\alpha}{}{\cdot}$ is continuous, the set $S = \{ x : y \in \reliablestat{\alpha}{}{x} \}$ is closed and because it's bounded this means that $x_0 = \inf S$ is actually a member of $S$. In other words, $y \in \reliablestat{\alpha}{}{x_0}$. It is easy to similarly show, using the continuity and invertibility of $\max \reliablestat{\alpha}{}{x_0}$, that in fact $y = \max \reliablestat{\alpha}{}{x_0}$.

To show that $\sup U \leq x_0 + d$, we now observe that since $\left| \interpretablestat{\alpha}{}{y} \right| = d$, we know that $y \notin \reliablestat{\alpha}{}{x}$ for all $x > x_0 + d$. This is either because $y > \max \reliablestat{\alpha}{}{x}$ or because $y < \min \reliablestat{\alpha}{}{x}$. However, since $y \in \reliablestat{\alpha}{}{x_0}$ and $\max \reliablestat{\alpha}{}{\cdot}$ is an increasing function, no $x > x_0$ can have $y > \max \reliablestat{\alpha}{}{x}$. Thus the only option remaining is that $y < \min \reliablestat{\alpha}{}{x}$, which gives that if $D_x$ is a sample of size $n$ from any $\mcZ \in \Q$ with $\Phi(\mcZ) = x > x_0 + d$, we will have
\[ \Pr{\hvphi(D_x) < y} < \alpha \]
But, as we've shown, the critical value of the test in question is $\max \reliablestat{\alpha}{}{x_0}$, which equals $y$. We therefore have that
\[ \powerfunc{\alpha}{x_0}(x) = \Pr{\hvphi(D_x) \geq y} \geq 1- \alpha \]
which implies that $x$ is not contained in $U$, as desired.

\proofwaypoint{Proof of the second statement} We again let $U$ denote the uncertain set of $\powerfunc{\alpha}{x_0}$. What are the infimum and supremum of $U$? To answer this, we note once again that $\alpha < 1/2$ implies that $\inf U = x_0$ and moreover, since $|U| = d$, we also have that $\sup U = x_0 + d$.

To prove our claim, we will establish that $\min \interpretablestat{\alpha}{}{y} = x_0$ and that $\max \interpretablestat{\alpha}{}{y} = x_0 + d$. The fact that $\min \interpretablestat{\alpha}{}{y} = x_0$ follows easily from $y = \max \reliablestat{\alpha}{}{x_0}$ and the fact that $\max \reliablestat{\alpha}{}{\cdot}$ is an increasing function. It is therefore left only to show the latter claim.

To establish that $\max \interpretablestat{\alpha}{}{y} = x_0 + d$, let us first show that $\max \interpretablestat{\alpha}{}{y} \geq x_0 + d$. We know that since $\sup U = x_0 + d$, we can find $x$ arbitrarily close to $x_0 + d$ from below such that $\powerfunc{\alpha}{x_0}(x) < 1-\alpha$. Again, since the critical value of the test in question is $y$, this means that there exists some $\mcZ \in \Q$ with $\Phi(\mcZ) = x$ such that if $D_x$ is a sample of size $n$ from $\mcZ$ then
\[ \powerfunc{\alpha}{x_0}(x) = \Pr{\hvphi(D_x) \geq y} < 1- \alpha \]
i.e.,
\[ \Pr{\hvphi(D) < y} \geq \alpha \]
and so $y \in \reliablestat{\alpha}{}{x}$. This means that $x \in \interpretablestat{\alpha}{}{y}$.

To show that $\max \interpretablestat{\alpha}{}{y} \leq x_0 + d$, we observe that for $x > x_0 + d$ we must have $\powerfunc{\alpha}{x_0}(x) \geq 1-\alpha$. Since the critical value of the test $T^{x_0}_\alpha(x)$ is $y$, this implies that if $D_x$ is a sample of size $n$ from any $\mcZ \in \Q$ with $\Phi(\mcZ) = x$, then
\[ \powerfunc{\alpha}{x_0}(x) = \Pr{\hvphi(D_x) \geq y} \geq 1 - \alpha \]
i.e.,
\[ \Pr{\hvphi(D) < y} < \alpha \]
In other words, $y \notin \reliablestat{\alpha}{}{x}$ for any $x > x_0 + d$, as desired.
\end{proof}

\subsection{Discussion}
What does the above result tell us about equitability? The first consequence of it is the following formal definition of equitability/interpretability in terms of statistical power, which we present without proof.

\begin{thm}
Fix a set $\Q \subset \P$, and a function $\Phi : \Q \rightarrow [0,1]$. Let $\hvphi$ be a statistic with the property that $\max \reliablestat{\alpha}{}{x}$ is a continuous increasing function of $x$, and fix some $\alpha \in [0, 1/2]$ and some $d > 0$. Then the following are equivalent:
\begin{enumerate}	
\item $\hvphi$ is a worst-case $1/d$-interpretable proxy for $\Phi$ with confidence $1-2\alpha$.
\item For every $x_0, x_1 \in [0,1]$ satisfying $x_1 - x_0 \geq d$, there exists a right-tailed test based on $\hvphi$ that can distinguish between $H_0 : \Phi(\mcZ) = x_0$ and $H_1 : \Phi(\mcZ) = x_1 $ with type I error at most $\alpha$ and power at least $1-\alpha$.
\end{enumerate}
\end{thm}

This definition shows what the concept of equitability/interpretability is fundamentally about: being able to distinguish not just signal ($\Phi > 0$) from no signal ($\Phi = 0$) but also stronger signal ($\Phi = x_1$) from weaker signal ($\Phi = x_0$). This is the essence of the difference between equitability/interpretability and power against statistical independence.

The definition also shows that equitability and intepretability --- to the extent they can be achieved --- subsume power against independence. To see this, suppose again that $\Phi(\mcZ) = 0$ exactly when $\mcZ$ exhibits statistical independence. By setting $x_0 = 0$ in the definition, we obtain the following corollary.
\begin{cor}
Fix a set $\Q \subset \P$, a function $\Phi : \Q \rightarrow [0,1]$ such that $\Phi(\mcZ) = 0$ iff $\mcZ$ exhibits statistical independence, and some $\alpha \in [0,1/2]$. Let $\hvphi$ be a worst-case $1/d$-interpretable proxy for $\Phi$ with confidence $1-2\alpha$, and assume that $\max \reliablestat{\alpha}{}{\cdot}$ is a continuous increasing function. The power of the right-tailed test based on $\hvphi$ at distinguishing $H_1 : \Phi(\mcZ) = d$ from statistical independence with type I error at most $\alpha$ is at least $1 - \alpha$.
\end{cor}
In other words, equitability/interpretability implies power against independence. However, equitability/interpretability is actually a stronger requirement: as the theorem shows, to be interpretable a statistic must yield a right-tailed test that is well-powered not only to detect deviations from independence ($\Phi(\mcZ) = 0$) but also from any fixed level of ``noisiness" (e.g., $\Phi(\mcZ) = 0.3$). This indeed makes sense when a data set contains an overwhelming number of relationships that exhibit, say $\Phi(\mcZ) = 0.3$ and that we would like to ignore because they are not as interesting as the small number of relationships with $\Phi(\mcZ) = 0.8$.

It is our hope that by formalizing the relationship between equitability and power against independence, our equivalence result will clarify the differences between these two properties, thereby addressing some of the concerns raised about the power of $\MIC$ against statistical independence (\cite{simon2012comment} and~\cite{gorfine2012comment}). We of course do agree that power against independence is a very important goal that is often the right one, and if all other things are equal more power is certainly always better. To this end, we have worked to greatly enhance $\MIC$'s power, both through better choice of parameters and through use of the estimators introduced later in this paper, to the point where it is often competitive with the state of the art. (The results of this work are forthcoming in the companion paper.) However, we also think that limiting one's analysis of $\MIC$ to power against statistical independence alone is not the right way to think about its utility.

For example, in~\cite{simon2012comment}, Simon and Tibshirani write ``The `equitability' property of MIC is not very useful, if it has low power". However, as the result described above shows, the question is ``power against what?". If one is interested only in power against statistical independence (e.g. $R^2 = 0$, in the setting of functional relationships), then choosing a statistic based solely on this property is the correct way to proceed. However, when the relationships in a dataset that exhibit non-trivial statistical dependence number in the hundreds of thousands, it often becomes necessary to be more stringent in deciding which of them to manually examine. As our result in this section shows, this can be thought of as defining one's null hypothesis to be $R^2 \leq \alpha$ for some $\alpha > 0$. In such a case, the statistic is not being used to identify {\em any} instance of dependence, but rather to identify any instance of dependence {\em of a certain minimal strength}. In other words, when used on relationships in $\Q$, an equitable statistic is a measure of effect size rather than a statistical test, and as with other measures of effect size, analyzing its power against only one null hypothesis (that of statistical independence alone) is therefore inappropriate.

Of course, when the relationships being sought in a dataset are expected to be very noisy, the above paradigm does not make sense and it is quite reasonable to ignore equitability and seek a statistic that maximizes power specifically against statistical independence. This issue, along with a broader discussion of when equitability is an appropriate desideratum, is discussed in more detail in the upcoming companion paper. From a theoretical standpoint, our result here simply formalizes the notion that these concepts, while distinct, are related, and shows that the former --- to the extent that it can be achieved --- implies the latter.

\section{$\MIC$ and the MINE statistics as consistent estimators}
\label{sec:consistency}
$\MIC$ is defined as the maximal element of a matrix called the characteristic matrix. However, both of these quantities are defined in~\cite{MINE} as statistics rather than as properties of distributions that can then be estimated from samples. Here we define the quantities that these two statistics turn out to estimate, and we prove that they do so. Thinking about these statistic as consistent estimators and then analyzing their behavior in the infinite-data limit subsumes and strengthens several previous results about $\MIC$, gives a better interpretation of the parameters in the definition of $\MIC$, clarifies the relationship of $\MIC$ to other measures of dependence, especially mutual information, and allows us to introduce new, better estimators that have improved performance.

In this section, we focus on introducing the population value of $\MIC$ (which we will call $\popMIC$) and proving that $\MIC$ is a consistent estimator of it, and then give a discussion of some immediate consequences of this approach. Subsequent sections of the paper are devoted to analyzing $\popMIC$ and stating new estimators of it.

\subsection{Definitions}
We begin by defining the characteristic matrix as a property of the distribution of two jointly distributed random variables $(X,Y)$ rather than as a statistic. In the sequel we will use $G(k,\ell)$ to denote, for positive integers $k$ and $\ell$, the set of all $k$-by-$\ell$ grids (possibly with empty rows/columns).

\begin{definition}
\label{def:charmatrix}
Let $(X,Y)$ be jointly distributed random variables on $[0,1] \times [0,1]$. For a grid $G$, let $(X,Y)|_G = (\mbox{col}_G(X), \mbox{row}_G(Y))$ where $\mbox{col}_G(X)$ is the column of $G$ containing $X$ and $\mbox{row}_G(Y)$ is analogously defined. Let
\[ I^*((X, Y), k, \ell) = \max_{G \in G(k,\ell)} I((X,Y)|_G) \]
where $I(X,Y)$ represents the mutual information of $X$ and $Y$.
The {\em population characteristic matrix} of $(X,Y)$, denoted by $M(X,Y)$, is defined by
\[ M(X,Y)_{k,\ell} = \frac{I^*((X,Y),k,\ell)}{\log \min \{k, \ell\}} \]
for $k, \ell > 1$.
\end{definition}
Note that in the above definition, $I$ refers to mutual information (see, e.g., \cite{csiszar2004information} and \cite{csiszar2008axiomatic}), not to an interpretable interval as in the previous sections.

The characteristic matrix is so named because in \cite{MINE} it was hypothesized that this matrix has a characteristic shape for different relationship types, such that different properties of this matrix may correspond to different properties of relationships. One such property was the maximal value of the matrix. This is called the maximal information coefficient (MIC), and is defined below.

\begin{definition}
Let $(X,Y)$ be jointly distributed random variables on $[0,1] \times [0,1]$. The {\em population maximal information coefficient} ($\popMIC$) of $(X, Y)$ is defined by
\[ \popMIC(X, Y) = \sup M(X,Y) \]
\end{definition}

We now define the corresponding statistics introduced in \cite{MINE}.

\begin{rmk}
In the rest of this paper, we will sometimes have a sample $D$ from the distribution of $(X,Y)$ rather than the distribution itself. We will abuse notation by using $D$ to refer both to the set of points that is the sample, as well as to the uniform distribution over those points. In the latter case, it will then make sense to talk about $I^*(D, k, \ell)$, as we are about to do below.
\end{rmk}

\begin{definition}
Let $D \subset [0,1] \times [0,1]$ be a set of ordered pairs. Given a function $B : \Z^+ \rightarrow \Z^+$, we define the {\em sample characteristic matrix} of $D$ to be
\[ \hat{M}_B(D)_{k,\ell} = \left\{
        \begin{array}{ll}
            \frac{I^*(D, k, \ell)}{ \log \min \{k, \ell \}} & \quad k\ell \leq B(|D|) \\
            0 & \quad k\ell > B(|D|)
        \end{array}
    \right. \]
\end{definition}

\begin{definition}
Let $D \subset [0,1] \times [0,1]$ be a set of ordered pairs, and let $B : \Z^+ \rightarrow \Z^+$. We define
\[ \MIC_B(D) = \max \hat{M}_B(D) \]
\end{definition}

In \cite{MINE}, other characteristic matrix properties were introduced as well (e.g. maximum asymmetry score [MAS], maximum edge value [MEV], etc.). These can be analogously presented as functions of random variables together with a corresponding statistic for each property.

\subsection{The main consistency result}
We now show that the statistic $\MIC$ defined above is in fact a consistent estimator of $\popMIC$. This is a consequence of the following more general result, which will be the main theorem of this section. In the theorem statement below, we let $m^\infty$ be the space of infinite matrices equipped with the supremum norm, and we let $r_i : m^\infty \rightarrow m^\infty$ denote the projection
\[ r_i(A)_{k,\ell} = \left\{
        \begin{array}{ll}
            A_{k,\ell} & \quad k\ell \leq i \\
            0 & \quad k\ell > i
        \end{array}
    \right. \]
\begin{thm*}
Let $f : m^\infty \rightarrow \R$ be uniformly continuous, and assume that $f \circ r_i \rightarrow f$ pointwise. Then for every random variable $(X,Y)$ supported on $[0,1] \times [0,1]$, the statistic $f(\hat{M}_B(\cdot))$ is a consistent estimator of $f(M(X,Y))$ provided $\omega(1) < B(n) \leq O(n^{1-\ep})$ for $\ep > 0$.
\end{thm*}

Since the supremum of a matrix is uniformly continuous as a function on $m^\infty$ and can be realized as the limit of maxima of larger and larger segments of the matrix, this theorem gives us the following corollary.
\begin{cor}
$\MIC_B$ is a consistent estimator of $\popMIC$ provided $\omega(1) < B(n) \leq O(n^{1-\ep})$ for $\ep > 0$.
\end{cor}
It is easily verified that analogous corollaries also hold for the statistics $\widehat{\mbox{MAS}}$ and $\widehat{\mbox{MEV}}$ defined in \cite{MINE} (and referred to there simply as MAS and MEV). Interestingly, it is unclear a priori whether such a result exists for MCN, since that statistic is not a uniformly continuous function of the sample characteristic matrix.

Before we prove this theorem, we will first give some intuition for why it should hold, and also for why it is non-trivial to prove. We then present the general strategy for the proof before giving the proof itself.

\subsubsection{Intuition}
Fix a random variable $(X,Y)$ and let $D$ be a sample of size $n$ from its distribution. It is known that, for a fixed grid $G$, $I(D|_G)$ is a consistent estimator of $I((X,Y)|_G)$ \cite{roulston1999estimating}. We might therefore expect $I^*(D, k, \ell)$ to be a consistent estimator of $I^*((X, Y), k, \ell)$ as well. And if $I^*(D, k, \ell)$ is a consistent estimator of $I^*((X, Y), k, \ell)$, then we might expect the maximum of the sample characteristic matrix (which just consists of normalized $I^*$ terms) to be a consistent estimator of the supremum of the true characteristic matrix.

These intuitions turn out to be true, but there are two reasons they are non-trivial to prove. First, consistency for $I^*$ does not follow from abstract considerations since the maximum of an infinite set of estimators is not necessarily a consistent estimator of the supremum of the estimands\footnote{If $\hvphi_1, \ldots, \hvphi_k$ is a finite set of estimators, then a union bound shows that the random variable $(\hvphi_1(D), \ldots, \hvphi_k(D))$ converges in probability to $(\varphi_1, \ldots, \varphi_k)$ with respect to the supremum metric. The continuous mapping theorem then gives the desired result. However, if the set of estimators is infinite, the union bound cannot be employed. And indeed, if we let $\varphi_1 = \cdots = \varphi_k = 0$, let $D$ represent a sample of size $n$, and suppose that $\hvphi_i(D)$ has all of its probability mass on the point $i/n$, then each $\hvphi_i$ is consistent but their supremum is always infinite.}. Second, consistency of $I^*$ alone does not suffice to show that the maximum of the sample characteristic matrix converges to $\popMIC$. In particular, if $B(n)$ grows too quickly, and the convergence of $I^*(D, k, \ell)$ to $I^*((X, Y), k, \ell)$ is slow, inflated values of $\MIC$ can result. To see this, notice that if $B(n) = \infty$ then $\MIC = 1$ always, even though each individual entry of the sample characteristic matrix converges to its true value eventually.

The technical heart of the proof is overcoming these obstacles by using the dependences between the quantities $I(D|_G)$ for different grids $G$ to not only show the consistency of $I^*(D, k, \ell)$ but then to quantify how quickly $I^*(D, k, \ell)$ actually converges to $I^*((X, Y), k, \ell)$.

\subsubsection{Proof strategy}
We will prove the theorem by a sequence of lemmas that build on each other to bound the bias of $I^*(D, k, \ell)$. The general strategy is to capture the dependencies between different $k$-by-$\ell$ grids $G$ by considering a ``master grid" $\Gamma$ that contains many more than $k\ell$ cells. Given this master grid, we first bound the difference between $I(D|_G)$ and $I((X,Y)|_G)$ only for sub-grids $G$ of $\Gamma$. The bound will be in terms of the difference between $D|_\Gamma$ and $(X,Y)|_\Gamma$. We then show that this bound can be extended without too much loss to all $k$-by-$\ell$ grids. This will give us what we seek, because then the differences between $I(D|_G)$ and $I((X,Y)|_G)$ will be uniformly bounded for all grids $G$ in terms of the same random variable: $D|_\Gamma$. Once this is done, standard arguments will give us the consistency we seek.

\subsubsection{The proof}
The proof of this result will sometimes require technical facts about entropy and mutual information that are self-contained and unrelated to the central idea behind our argument. These lemmas are consolidated in Appendix~\ref{appendix:technical}.

We begin by using one of these technical lemmas to prove a bound on the difference between $I(D|_G)$ and $I((X,Y)|_G)$ that is uniform over all grids $G$ that are sub-grids of a much denser grid $\Gamma$. The common structure imposed by $\Gamma$ will allow us to capture the dependence between the quantities $\left| I(D|_G) - I((X,Y)|_G) \right|$ for different grids $G$.

\begin{lemma}
\label{lem:boundMIOnSubgrids}
Let $\Pi = (\Pi_X, \Pi_Y)$ and $\Psi = (\Psi_X, \Psi_Y)$ be random variables distributed over the cells of a grid $\Gamma$, and let $(\pi_{i,j})$ and $(\psi_{i,j})$ be their respective distributions. Define
\[ \ep_{i,j} = \frac{\psi_{i,j} - \pi_{i, j}}{\pi_{i,j}} \]
Let $G$ be a sub-grid of $\Gamma$ with $B$ cells. Then, for every $0 < a < 1$ there exists some $A > 0$ such that
\[ \left| I(\Psi|_G) - I(\Pi|_G) \right| \leq \left( \log B \right) A \sum_{i,j} |\ep_{i,j}| \]
when $|\ep_{i,j}| \leq 1 - a$ for all $i,j$.
\end{lemma}

\begin{proof}
Let $P = \Pi|_G$ and $Q = \Psi|_G$ be the random variables induced by $\Pi$ and $\Psi$ respectively on the cells of $G$. Using the fact that $I(X,Y) = H(X) + H(Y) - H(X,Y)$, we write
\[ \left| I(Q) - I(P) \right| \leq \left| H(Q_X) - H(P_X) \right| + \left| H(Q_Y) - H(P_Y) \right| + \left| H(Q) - H(P) \right| \]
where $Q_X$ and $P_X$ denote the marginal distributions on the columns of $G$ and $Q_Y$ and $P_Y$ denote the marginal distributions on the rows. We can bound each of the above terms using a Taylor expansion argument given in Lemma~\ref{lem:boundEntropy}, whose proof is found in the appendix. Doing so gives
\[ (\ln B) \left( \sum_i \O{|\ep_{i,*}|} + \sum_j \O{|\ep_{*,j}|} + \sum_{i,j} \O{|\ep_{i,j}|} \right) \]
where
\[ \ep_{i,*} = \frac{\sum_j (\psi_{i,j} - \pi_{i,j})}{\sum_j \pi_{i,j}} \]
and $\ep_{*,j}$ is defined analogously.

To obtain the result, we observe that
\[ \left| \ep_{i,*} \right|
	= \left| \sum_j \frac{\pi_{i,j}}{\sum_j \pi_{i,j}} \ep_{i,j} \right|
	\leq \sum_j \frac{\pi_{i,j}}{\sum_j \pi_{i,j}} \left| \ep_{i,j} \right|
	\leq \sum_j \left| \ep_{i,j} \right|
\]
since $\pi_{i,j} / \sum_j \pi_{i,j} \leq 1$, and the analogous bound holds for $\left| \ep_{*,j} \right|$.
\end{proof}

We now extend Lemma~\ref{lem:boundMIOnSubgrids} to all grids with $B$ cells rather than just those that are sub-grids of the master grid $\Gamma$. It is useful at this point to recall that, given a distribution $(X,Y)$, an {\em equipartition} of $(X,Y)$ is a grid $G$ such that all the rows of $(X,Y)|_G$ have the same probability mass, and all the columns do as well.

\begin{lemma}
\label{lem:boundMI}
Let $\Pi = (\Pi_X, \Pi_Y)$ and $\Psi = (\Psi_X, \Psi_Y)$ be random variables distributed over $[0,1] \times [0,1]$, and let $\Gamma$ be a grid. Define $\ep_{i,j}$ on $\Pi|_\Gamma$ and $\Psi|_\Gamma$ as in Lemma~\ref{lem:boundMIOnSubgrids}. Let $G$ be any $k$-by-$\ell$ grid, and let $\delta$ (resp. $d$) represent the total probability mass of $\Pi|_\Gamma$ (resp. $\Psi|_\Gamma$) falling in cells of $\Gamma$ that are not contained in individual cells of $G$. We have that
\[ \left| I(\Psi|_G) - I(\Pi|_G) \right|
\leq (\log (4k\ell))\sum_{i,j} \O{ \left| \ep_{i,j} \right| }
	+ 2 \left( H_b(\delta) + H_b(d) + \delta + d \right)  \]
provided that the $|\ep_{i,j}|$ are bounded away from 1 and that $d, \delta \leq 1/2$.
\end{lemma}
\begin{proof}
In the proof below, we use the convention that for any two grids $G$ and $G'$ and any distribution $\mcZ$, the expression $\Delta^{\mcZ}(G, G')$ denotes $| I(\mcZ|_G) - I(\mcZ|_{G'})|$. In addition, we refer to any horizontal or vertical line in $G$ that is not in $\Gamma$ as a {\em dissonant} line of $G$.

Consider the grid $G'$ obtained by adding to $G$ the two lines in $\Gamma$ that surround each dissonant line of $G$ and then removing all the dissonant lines of $G$. This grid $G'$ is clearly a sub-grid of $\Gamma$. And in Lemma~\ref{lem:boundMIOfCloseGrids}, whose proof we defer to the appendix, we do some careful accounting to show that $G'$ has the property that
\[ \Delta^\Psi(G, G') \leq 2 \left( H_b(d) + d \right) \]
and
\[ \Delta^{\Phi}(G', G) \leq 2 \left( H_b(\delta) + \delta \right) \]
We can then bound $|I(\Psi|_G) - I(\Phi|_G)|$ using the triangle inequality by comparing it with
\[ \Delta^\Psi(G, G')
+ \left| I \left( \Psi|_{G'} \right) - I \left( \Phi|_{G'} \right) \right|
+ \Delta^{\Phi}(G', G) \]
and bounding the middle term using Lemma~\ref{lem:boundMIOnSubgrids}.
\end{proof}

We now use the fact that the variables $\ep_{i,j}$ defined in Lemma~\ref{lem:boundMIOnSubgrids} are small with high probability to give a concrete bound on the bias of $I(D|_G)$ that is uniform over all $k$-by-$\ell$ grids $G$ and that holds with high probability.

\begin{lemma}
\label{lem:unionBoundOverGrids}
Let $(X,Y)$ be a continuous random variable, and let $D_n$ represent a random sample of size $n$ from the distribution of $(X,Y)$. For any $\alpha \geq 0$, any $\ep > 0$, and any integers $k, \ell > 1$, we have that for all $n$
\[ \left| I(D_n|_G) - I((X,Y)|_G) \right|
	\leq \O{ \frac{\log (4 k \ell)}{C(n)^\alpha} } + \O{ \frac{1}{n^{\ep/8}} } \]
for every $k$-by-$\ell$ grid $G$ with probability at least $1 - C(n) e^{-\Omega(n/C(n)^{1+2\alpha})}$, where $C(n) = k\ell n^{\ep/2}$.
\end{lemma}

\begin{proof}
Fix $n$, and let $\Gamma$ be an equipartition of the support of $(X,Y)$ into $kn^{\ep/4}$ rows and $\ell n^{\ep/4}$ columns. $C(n)$ is now the number of cells in $\Gamma$. Lemma~\ref{lem:boundMI}, with $\Pi = (X,Y)$ and $\Psi = D$, shows that $\left| I(D|_G) - I((X,Y)|_G) \right|$ is at most
\[ \left( \log (4k\ell) \right) \sum_i \O{|\ep_{i,j}|}
	+ 2 \left( H_b(\delta) + H_b(d) + \delta + d \right) \]
provided the $\ep_{i,j}$ have absolute value bounded away from 1, and provided that $d, \delta \leq 1/2$.

The remainder of the proof proceeds as follows. We first show that the $\ep_{i,j}$ are small with high probability. This will both show that the lemma's requirement on the $\ep_{i,j}$ holds and allow us to bound the sum in the inequality above. We will then use our bound on the $\ep_{i,j}$ to bound $d$ in terms of $\delta$. Finally, we will bound $\delta$ using the fact that the number of rows and columns in $\Gamma$ increases with $n$. This will give us that $d, \delta \leq 1/2$ and allow us to bound the rest of the terms in the expression above.

\proofwaypoint{Bounding the $\ep_{i,j}$}
We bound the $\ep_{i,j}$ using a multiplicative Chernoff bound. Let $\pi_{i,j}$ and $\psi_{i,j}$ represent the probability mass functions of $(X,Y)|_\Gamma$ and $D|_\Gamma$ respectively. We write
\begin{eqnarray*}
\Pr{ \left| \ep_{i,j} \right| \geq \delta} &=& \Pr{\pi_{i,j}(1 - \delta) \leq \psi_{i,j} \leq \pi_{i,j}(1 + \delta)} \\
	&\leq& e^{-\Omega(n\pi_{i,j} \delta^2)}
\end{eqnarray*}
since $\psi_{i,j}$ is a sum of $n$ i.i.d Bernoulli random variables and $\E{\psi_{i,j}} = n\pi_{i,j}$. (See, e.g., \cite{mitzenmacher2005probability}.) Setting $\delta = \sqrt{\pi_{i,j}}/C(n)^{1/2 + \alpha}$ then gives
\[ \Pr{|\ep_{i,j}| \geq \frac{\sqrt{\pi_{i,j}}}{C(n)^{1/2 + \alpha}}} \leq e^{-\Omega(n/C(n)^{1+2\alpha})}\]
A union bound over the pairs $(i,j)$ then gives that, with the desired probability, the above bound on $|\ep_{i,j}|$ holds for all $i,j$.

\proofwaypoint{Bounding $\sum \O{ \left| \ep_{i,j} \right| }$}
The above bound on the $\ep_{i,j}$ implies that
\begin{eqnarray*}
\sum_i \O{|\ep_{i,j}|} &\leq& \O{\frac{1}{C(n)^{1/2 + \alpha}} } \sum_{i,j} \sqrt{\pi_{i,j}} \\
	&\leq& \O{ \frac{1}{C(n)^{1/2 + \alpha}} } \sqrt{C(n)} \\
	&\leq& \O{ \frac{1}{C(n)^{\alpha}} }
\end{eqnarray*}
where the second line follows from the fact that the function $\sum \sqrt{\pi_{i,j}}$ is symmetric and concave and therefore, when restricted to the hyperplane $\sum \pi_{i,j }= 1$, must achieve its maximum when $\pi_{i,j} = 1/C(n)$ for all $i,j$.

\proofwaypoint{Bounding $d$ in terms of $\delta$}
We use our bound on the $\ep_{i,j}$ to bound $d$. We do so by observing that it implies
\[
\psi_{i,j} \leq \pi_{i,j} \left( 1 + \frac{\sqrt{\pi_{i,j}}}{C(n)^{1/2 + \alpha}} \right)
	= \pi_{i,j} + \frac{\pi_{i,j}^{3/2}}{C(n)^{1/2 + \alpha}}
	\leq \pi_{i,j} + \frac{\pi_{i,j}}{C(n)^{1/2 + \alpha}}
	\leq 2\pi_{i,j} \]
since $\pi_{i,j} \leq 1$ and $C(n) \geq 1$.

The connection to $d$ comes from the fact that for any column $j$ of $\Gamma$, this means that
\[ \psi_{*,j} = \sum_i \psi_{i,j} \leq 2 \sum_i \pi_{i,j} = 2 \pi_{*,j} \]
This also applies to the sums across rows. Since $d$ is a sum of terms of the form $\psi_{*,j}$ and $\psi_{i,*}$ for $j$ in some index set $J$ and $i$ in an index set $I$, and $\delta$ is a sum of terms of the form $\pi_{*,j}$ and $\pi_{i,*}$ with the same index sets, we therefore get that $d \leq 2 \delta$.

\proofwaypoint{Bounding $\delta$ and obtaining the result}
To bound $\delta$, we observe that because $G$ has at most $\ell - 1$ vertical lines and $k -1$ horizontal lines, we have
\[ \delta \leq \frac{\ell}{\ell n^{\ep/ 4}} + \frac{k}{k n^{\ep/4}} \leq \frac{2}{n^{\ep/4}} \]

This bound on $\delta$ allows us to bound the terms involving $d$ and $\delta$ by
\[ H_b(\delta) + H_b(d) + \delta + d \leq H_b\left( \frac{2}{n^{\ep/4}} \right) + H_b\left( \frac{4}{n^{\ep/4}} \right) + \O{\frac{1}{n^{\ep/4}}} \]
Finally, we observe that around 0, $H_b(x) = O(x^a)$ for every $0 \leq a < 1$. Taking $a = 1/2$ gives the result.
\end{proof}

Our final lemma shows that as long as $B(n)$ doesn't grow too fast, the bound from the previous lemma yields a uniform bound on the entire sample characteristic matrix. This is done by specifying an error threshold for which Lemma~\ref{lem:unionBoundOverGrids} gives us a bound that holds with high probability, and then invoking a union bound.

\begin{lemma}
\label{lem:boundCharMatrixEntries}
Fix a continuous random variable $(X,Y)$, and let $D_n$ be $n$ iid samples from it. Let
\[ \hat{M}(D_n)_{k, \ell} = \frac{I^*(D_n, k, \ell)}{\log \min \{k, \ell \} } \]
Then for every $B(n) = \O{n^{1-\ep}}$, there exists an $\alpha > 0$ such that for sufficiently large $n$,
\[
\left| \hat{M}(D_n)_{k, \ell} - M_{k, \ell} \right|
	\leq \O{ \frac{\log n}{n^{\ep \alpha/2}} }
		+ \O{ \frac{1}{n^{\ep / 8}} }
\]
holds for all $k\ell \leq B(n)$ with probability $P(n) = 1 - o(1)$, where $M_{k,\ell}$ is the $k, \ell$-th entry of the characteristic matrix of $(X,Y)$.
\end{lemma}

\begin{proof}
We show that any $0 < \alpha < \ep/(4 - 2\ep)$ suffices.

Fix $k, \ell$. Lemma~\ref{lem:unionBoundOverGrids} implies that with high probability the difference $|\hat{M}(D_n)_{k, \ell} - M_{k,\ell}|$ is at most
\begin{eqnarray*}
\O{ \frac{\log(4k\ell)}{C(n)^\alpha} } + \O{ \frac{1}{n^{\ep / 8}} }
	&\leq& \O{ \frac{\log(n)}{C(n)^\alpha} } + \O{ \frac{1}{n^{\ep / 8}} } \\
&\leq& \O{ \frac{\log n}{n^{\alpha \ep/2}} } + \O{ \frac{1}{n^{\ep / 8}} }
\end{eqnarray*}
where the first inequality comes from $k\ell \leq B(n)$ and second is because $C(n) = k\ell n^{\ep/2} \geq n^{\ep/2}$. The lemma further states that the probability this holds is at least
\[ 1 - C(n)e^{-\Omega(n/C(n)^{1+2\alpha})} \geq 1 - \O{n}e^{-\Omega(n^a)} \]
for some positive $a$. This is because $C(n) \leq B(n)n^{\ep/2} \leq \O{n^{1 - \ep/2}}$ for large $n$, and so our choice of $\alpha$ ensures that $C(n)^{1 + 2\alpha} = \O{n^{1-a}}$ for some $a > 0$.

We can then perform a union bound over all pairs $k\ell \leq B(n)$: since the number of such pairs can be bounded by a  polynomial in $n$, we have that the desired condition is satisfied for all $k\ell \leq B(n)$ with probability approaching 1.
\end{proof}

We are now ready to prove our main result. Recall that $m^\infty$ is the space of infinite matrices equipped with the supremum norm, and the projection $r_i : m^\infty \rightarrow m^\infty$ is defined by
\[ r_i(A)_{k,\ell} = \left\{
        \begin{array}{ll}
            A_{k,\ell} & \quad k\ell \leq i \\
            0 & \quad k\ell > i
        \end{array}
    \right. \]

\begin{thm}
\label{thm:estimator}
Let $f : m^\infty \rightarrow \R$ be uniformly continuous, and assume that $f \circ r_i \rightarrow f$ pointwise. Then for every random variable $(X,Y)$ supported on $[0,1] \times [0,1]$, the statistic $f(\hat{M}_B(\cdot))$ is a consistent estimator of $f(M(X,Y))$ provided $\omega(1) < B(n) \leq O(n^{1-\ep})$ for $\ep > 0$.
\end{thm}

\begin{proof}
Let the random variable $D_n$ denote $n$ iid samples from $(X,Y)$. We wish to show that $f(\hat{M}_B(D_n))$ converges in probability to $f(M(X,Y))$.

Let $M_i = r_i(M)$ and let $N$ denote $B(n)$. We begin by writing
\begin{eqnarray*}
\left| f \left( \hat{M}_B(D_n) \right) - f(M) \right| &\leq& \left| f \left( \hat{M}_B(D_n) \right) - f \left( M_N \right) \right| + \left| f \left( M_N \right) - f(M) \right| \\
	&=& \left| f \left( \hat{M}_B(D_n) \right) - f \left( M_N \right) \right| + \left| f(r_N(M)) - f(M) \right|
\end{eqnarray*}
and observing that as $n \rightarrow \infty$, the second term vanishes by the pointwise convergence of $f \circ r_i$ and the fact that $B(n) > \omega(1)$. It therefore suffices to show that the first term converges to 0 in probability. Since $f$ is uniformly continuous, we can establish this via a simple adaptation of the continuous mapping theorem, which says that if the sequence of random variables $R_n \rightarrow R$ in probability, and $g$ is continuous, then $g(R_n) \rightarrow g(R)$ in probability. We will replace $R$ with a second sequence, and replace continuity with uniform continuity.

Let $\| \cdot \|$ denote the supremum norm on $m^\infty$, and fix any $z > 0$. Then, for any $\delta > 0$, define
\[ C_{\delta} = \left\{ A \in m^\infty : \exists A' \in m^\infty\mbox{ s.t. } \| A - A' \| < \delta, \left| f(A) - fA') \right| > z \right\} \]
This is the set of matrices $A \in m^\infty$ for which it is possible to find, within a $\delta$-neighborhood of $A$, a second matrix that $f$ maps to more than $z$ away from $f(A)$. Because $f$ is uniformly continuous, there exists a $\delta^*$ sufficiently small so that $C_{\delta^*} = \emptyset$.

Suppose that $| f(\hat{M}_B(D_n)) - f(M_N) | > z$. This means that either $\| \hat{M}_B(D_n) - M_N \| > \delta^*$, or $M_N \in C_{\delta^*}$. The latter option is impossible since $C_{\delta^*} = \emptyset$, and Lemma~\ref{lem:boundCharMatrixEntries} tells us that $\Pr{\| \hat{M}_B(D_n) - M_N \| > \delta^*} \rightarrow 0$ as $n$ grows. We therefore have that
\[ \left| f \left( \hat{M}_B(D_n) \right) - f(M_N) \right| \rightarrow 0 \]
in probability, as desired.
\end{proof}

\subsection{Discussion: the relationship between $\popMIC$ and mutual information}
There are several advantages to the interpretation of $\MIC$ as a consistent estimator of $\popMIC$. First, it helps in understanding what in the statistics defined in \cite{MINE} is essential and what is ancillary. In particular, it is now clear that the choice of $B(n)$ does not change the actual quantity being estimated, something that while intuitive was not rigorously established previously. Thus, while the question of choosing $B(n)$ for practical performance still merits further investigation, we now rest assured that asymptotically the different choices are all equivalent, and that the choice of $B(n)$ represents a bias versus variance tradeoff in the estimator.

It is also now easier to see the relationship of $\MIC$ and $\popMIC$ to mutual information. We can do so through the following well known alternate definition of mutual information, which we state without proof.
\begin{prop}
\label{prop:mutualInfoSupremumDef}
For every pair of jointly distributed random variables $(X,Y)$, we have
\[ I(X, Y) = \sup_G I((X,Y)|_G) \]
\end{prop}
We can re-interpret this alternate definition of mutual information as follows
\begin{cor}
For every pair of jointly distributed random variables $(X,Y)$, we have
\[ I(X,Y) = \sup_{k, \ell > 1} I^*((X,Y), k, \ell) \]
\end{cor}
Comparison with the definition of $\popMIC$ now shows the relationship between $\popMIC$ and $I$. A simple normalization of $I$ would subject all the grids in Proposition \ref{prop:mutualInfoSupremumDef} to the same normalization; instead, $\popMIC$ subjects different sets of grids to different normalizations. Thus, $\popMIC$ is a penalized version of $I$ in which simple grids can take precedence over complex grids if they do well enough. As we will show in the next section, it actually turns out that this penalization can also be thought of as the ``minimal" smoothing operation necessary to make $\popMIC$ uniformly continuous. ($I$ itself is not continuous.)

Our consistency results, together with properties of $\popMIC$ that are easy to prove, subsume most of the theoretical results proven in \cite{MINE}. In particular, the results about functional relationships and superpositions of functional relationships achieving perfect scores and statistically independent variables achieving vanishing scores are now all just corollaries of our main theorem. Furthermore, the result that $\MIC \rightarrow 0$ in the case of statistical independence was only proven under certain restrictions on the space of grids that could be maximized over. Our main theorem now ensures that that result holds without these restrictions.

Lastly, the fact that $\MIC$ is simply a consistent estimator of $\popMIC$ allows us to ask whether there exist other, more efficient estimators of $\popMIC$. Over the next few sections we will develop the theory necessary to introduce new statistics for estimating $\popMIC$ that have better bias/variance properties, as well as better equitability, power against independence, and runtime. The first step to doing this is to prove some properties of $\popMIC$ and provide an alternate characterization of it. This is our task in the next section.

\section{The continuity of $\popMIC$ and the population characteristic matrix}
In the previous section we defined $\popMIC(X,Y)$, the maximal information coefficient of the jointly distributed random variables $(X,Y)$, to be the supremum of a matrix $M(X,Y)$ which we called the population characteristic matrix. Here we prove that when considered as functions of the probability density function (pdf) of $(X,Y)$, both $\popMIC(X,Y)$ and $M(X,Y)$ are continuous.

This is interesting for two reasons. First, it implies that computing the $\popMIC$ of a density estimated via a consistent density estimator is a legitimate way to estimate the $\popMIC$ of a sample. Second, it sheds new light on the nature of the normalization by $\log \min \{k, \ell\}$ in the definition of the characteristic matrix. This normalization turns out to be crucial to our proof, and we will in fact show that it is in some sense the ``minimal" one necessary for achieving continuity. Along the way we will see that mutual information (as well as the Linfoot correlation) is not a continuous function of the pdf. This suggests an additional way to think of $\popMIC$: as a canonically smoothed version of mutual information that is uniformly continuous.

\subsection{Proving continuity}
When $f$ is a pdf of a random variable $(X,Y)$, we will abuse notation by using $\popMIC(f)$ and $M(f)$ to denote $\popMIC(X,Y)$ and $M(X,Y)$. Also, for a set $S$, we will write $\mathcal{P}(S)$ to denote the space of pdf's supported on $S$, equipped with the $L^1$ norm.

We begin by observing that the family of maps consisting of applying any finite grid to a pdf $f$ is uniformly equicontinuous. The reason this holds is that applying a grid to $f$ corresponds to applying a deterministic function to the random variable described by $f$, and deterministic functions cannot increase statistical distance.

\begin{prop}
\label{prop:equicontinuityOfGridding}
Given a grid $G$ and a pdf $f \in \mathcal{P}([0,1]^2)$, define $f|_G \in \mathcal{P}(G)$ to be the density function over the cells of $G$ that results from integrating $f$ in each cell of $G$. Let $\mathbb{G}$ be the set of all finite grids. The family of maps $\{ f \mapsto \mcZ |_G : G \in \mathbb{G} \}$ is uniformly equicontinuous.
\end{prop}
\begin{proof}
To establish uniform equicontinuity, we need to show that, given some $\ep > 0$, we can choose $\delta$ in a way that does not depend on $G$ or on $f$. We can do this because if $f$ and $f'$ are the pdfs of the random variables $R$ and $R'$ respectively, then the $L^1$ distance between $f$ and $f'$ equals twice $\Delta(R, R')$, the total variation distance between $R$ and $R'$. Moreover, for a fixed grid $G$, the pdfs $f|_G$ and $f'|_G$ are the pdf's of the random variables $G(R)$ and $G(R')$ respectively, where the function $G$ is interpreted as mapping a sample to its row and column coordinates in $G$. Thus, since no deterministic function of $R$ and $R'$ can increase their statistical distance, we get
\[ \left| f|_G - f'|_G \right| = 2 \Delta \left( G(R), G(R')\right) \leq 2 \Delta \left( R, R' \right) = \left| f - f' \right| \]
where $| \cdot |$ denotes the $L^1$ norm. This is the desired result.
\end{proof}

At this point it is tempting to try to use continuity properties of discrete mutual information to obtain uniform continuity of the characteristic matrix. And indeed, this strategy does yield that each {\em individual} entry of the characteristic matrix is a uniformly continuous function. However, mutual information is only uniformly continuous for a fixed grid resolution, and to obtain continuity of the entire (infinite) characteristic matrix we need to make a statement about all grid resolutions simultaneously.

Given this issue, how is it possible that the entire characteristic matrix is a continuous function? The answer is that this is achieved by the normalization of the characteristic matrix. To see why, suppose we have a distribution over a $k$-by-$\ell$ grid and we are allowed to move around $\delta$ probability mass for some small $\delta$. The largest change in discrete mutual information that this can cause increases as we increase $k$ and $\ell$. However, it turns out that we can bound the extent of this ``non-uniformity": the proposition below shows that as we move the probability mass around, the discrete mutual information can grow only linearly in the amount of mass we move with the constant factor bounded by $O(\log \min \{k, \ell \})$. Because $\log \min \{k, \ell\}$ is the quantity by which we normalize each entry of the characteristic matrix, this will be exactly enough to make the normalized matrix continuous.

\begin{prop}
\label{prop:boundedVariationOfI}
Let $I_{k,\ell} : \mathcal{P}(\{1,\ldots,k\} \times \{1, \ldots, \ell\}) \rightarrow \R$ denote the discrete mutual information function on $k$-by-$\ell$ grids. For $0 < \delta \leq 1/4$, the maximal amount we can change $I_{k,\ell}$ by moving at most $\delta$ of the probability mass is
\[ 4H_b(2\delta) + 7 \delta \log \min \{k, \ell \} \]
\end{prop}
\begin{proof}
It suffices to bound only the maximal possible {\em increase} in mutual information, since if we have a decrease in going from distribution $A$ to distribution $B$ then we can consider $B$ to be the starting distribution instead.

Without loss of generality, assume $k \leq \ell$, so that $\log \min \{k, \ell\} = \log k$. Suppose we have a pair $(X,Y)$ of jointly distributed random variables, and we move around at most $\delta$ probability mass to arrive at a new pair $(X', Y')$. Using $I(X,Y) = H(Y) - H(Y | X)$, we write
\[ \left| I(X,Y) - I(X',Y') \right| \leq \left| H(Y) - H(Y') \right| + \left| H(Y | X) - H(Y' | X') \right| \]

We now use Lemma~\ref{lem:entropyBoundForChangedMass}, which relates statistical distance to changes in entropy and is proven in the appendix, to separately bound each of the terms on the right hand side. Straightforward application of the lemma to $| H(Y) - H(Y') |$ shows that it is at most $2 H_b(2 \delta) + 3\delta \log k$.

Bounding the term with the conditional entropies is more complicated. Let $p_x = \Pr{X = x}$, and let $p_x' = \Pr{X' = x}$. We have
\begin{eqnarray}
\left| H(Y | X) - H(Y' | X') \right| &=& \sum_x \left| p_x H(Y | X = x) - p_x' H(Y' | X' = x) \right| \nonumber \\
	&\leq& \sum_x \left( p_x \left| H(Y | X = x) - H(Y' | X' = x) \right| + \left| p_x' - p_x \right| H(Y' | X' = x) \right) \nonumber \\
	&=& \sum_x p_x \left| H(Y | X = x) - H(Y' | X' = x) \right| + \sum_x \left| p_x' - p_x \right| \log k \nonumber \\
	&\leq& \sum_x p_x \left| H(Y | X = x) - H(Y' | X' = x) \right| + \delta \log k \label{line:boundMutualInfoChange}
\end{eqnarray}
where the last line is because $\sum_x |p_x - p_x'| \leq \delta$ and $H(Y' | X' = x) \leq \log k$.

Now let $\delta_{x+}$ be the magnitude of all the probability mass entering any cell in column $x$, let $\delta_{x-}$ be the magnitude of all the probability mass leaving any cell in column $x$, and let $\delta_x = \delta_{x+} + \delta_{x-}$. Using this notation, we can again apply Lemma~\ref{lem:entropyBoundForChangedMass} to obtain
\begin{eqnarray*}
\sum_x p_x \left| H(Y | X = x) - H(Y' | X' = x) \right| &\leq& \sum_x p_x \left( 2H_b \left( \frac{2\delta_x}{p_x} \right) + 3 \frac{\delta_x}{p_x} \log k \right) \\
	&=& 2\sum_x p_x H_b \left( \frac{2\delta_x}{p_x} \right) + 3 \sum_x \delta_x \log k \\
	&\leq& 2\sum_x p_x H_b \left( \frac{2\delta_x}{p_x} \right) + 3 \delta \log k \\
	&\leq& 2 H_b(2 \delta) + 3 \delta \log k
\end{eqnarray*}
where the last line is by application of Lemma~\ref{lem:boundWeightedAverageOfEntropy} from the appendix, which bounds weighted sums of binary entropies and was used in Section~\ref{sec:consistency} as well.

Combining this with Line~\eqref{line:boundMutualInfoChange} gives that
\[ \left| H(Y | X) - H(Y' | X') \right| \leq 2H_b(2\delta) + 4\delta \log k \]
which, together with the bound on $\left| H(Y) - H(Y') \right|$, gives the result.
\end{proof}

Having bounded the extent to which variation in mutual information depends on grid resolution, we are now ready to show the uniform continuity of the characteristic matrix.

\begin{thm}
The map $f \mapsto M(f)$ is uniformly continuous.
\end{thm}
\begin{proof}
We complete the proof in three steps. First, we show that a certain family of functions $F$ is uniformly equicontinuous. Second, we use this to show that a different family $F'$ consisting of functions of the form $\sup_{g \in A} g$ with $A \subset F$ is uniformly equicontinuous. Finally, we argue that since the entries of $M(f)$ consist of the functions in $F'$, this gives the result.

Define
\[ F = \left\{ f \mapsto \frac{I_{k,\ell}(f_G)}{\log \min \{k, \ell \} } : G \in \bigcup_{k, \ell \in \Z} G(k, \ell) \right\} \]
$F$ is uniformly equicontinuous by the following argument. Given some $\ep > 0$, we know (Proposition \ref{prop:equicontinuityOfGridding}) that restricting $f$ to an $\ep$-ball around any $f_0$ means that $f|_G$ will remain within an $\ep$-ball of $f_0|_G$ for any $G$. Proposition~\ref{prop:boundedVariationOfI} then tells us that if $\ep$ is sufficiently small then $I_{k,\ell}(\Phi|_G)$ will be at most
\[ 4H_b(2\ep) + 7\ep \log \min \{k,\ell\} \]
away from $I_{k,\ell}(\Phi_0|_G)$. After the normalization, this becomes at most $4H_b(2\ep) + 7\ep$, which goes to 0 (uniformly in $f$) as $\ep$ approaches 0, as desired.

Next, define 
\[ F' = \left\{ f \mapsto M(f)_{k,\ell} : k, \ell \in \Z_{> 1} \right\} \]
Each map in $F'$ is of the form $\sup_{g \in F'} g$ for some $A \subset F$. Therefore, for a given $\ep > 0$, whatever $\delta$ establishes the uniform equicontinuity for $F$ can be used to establish continuity of all the functions in $F'$. (To see this: $\sup_{g \in A} g$ can't increase by more than $\ep$ if no $g$ increases by more than $\ep$, and $\sup_{g \in A}g$ is also lower bounded by any of the $g$'s, so it can't decrease by more than $\ep$ either.) Since we can use the same $\delta$ for all of the maps in $F'$, they therefore form a uniformly equicontinuous family.

Finally, the $\delta$ provided by the uniform equicontinuity of $F'$ also ensures that $M(\Phi)$ is within $\ep$ of $M(\Phi_0)$ in the supremum norm, thus giving the uniform continuity of $f \mapsto M(f)$.
\end{proof}

\begin{cor}
The map $f \mapsto \popMIC(f)$ is uniformly continuous.
\end{cor}
Similar corollaries exist for any continuous function of the characteristic matrix (including the others introduced in \cite{MINE}).

\subsection{The normalization used in $M$ is necessary for continuity}
In the argument above we used the fact that each entry of the characteristic matrix is normalized to achieve its continuity as a function of the pdf. We now show that if we define the characteristic matrix with any smaller normalization, it will contain an infinite discontinuity when considered as a function of the pdf. 

\begin{prop}
Let $M_N$ be the characteristic matrix with a normalization, i.e.,
\[ M_N(X,Y) = \sup_{k,\ell} I^*((X,Y), k, \ell) / N(k, \ell) \]
If $N(k, \ell) = o(\log \min \{k,\ell\})$ along some path $P$, then $M_N$ and $\sup M_N$ are not continuous when considered as functions of $\mathcal{P}([0,1]^2)$.
\end{prop}
\begin{proof}
Let $M'$ be the characteristic matrix without any normalization, i.e.,
\[ M'(X,Y) = \sup_{k,\ell} I^*((X,Y), k, \ell) \]
and consider the random variable $Z$ that is uniformly distributed on $[0,1/2]^2$. Because $Z$ exhibits statistical independence, $M'(Z)$ is zero everywhere. Now define $Z_\ep$ to be uniformly distributed on $[0,1/2]^2$ with probability $1-\ep$ and uniformly distributed on the line from $(1/2, 1/2)$ to $(1,1)$ with probability $\ep$.

We lower-bound the $k,\ell$-th entry of $M'(Z_\ep)$. Without loss of generality suppose that $k \leq \ell$, and consider a grid that places all of $[0,1/2]^2$ into one cell and uniformly partitions the set $[1/2,1]^2$ into $k-1$ rows and $k-1$ columns. By considering just the rows/columns in the set $[1/2,1]^2$ we see that this grid gives a mutual information of at least $\ep \log(k-1)$. Thus, we have that for all $k, \ell$,
\[ M'(X,Y) \geq \ep \log \min \{k-1, \ell - 1\} \]

This implies that the limit of $M_N(Z_\ep)$ along $P$ will be $\infty$, and so the distance between $M_N(Z)$ and $M_N(Z_\ep)$ in the supremum norm will be infinite.
\end{proof}

\begin{cor}
Mutual information and the Linfoot correlation are not continuous as functions of $\mathcal{P}([0,1]^2)$.
\end{cor}
\begin{proof}
$I$ is the supremum of $M_N$ with $N(k,\ell) = 1$. The claim for the Linfoot correlation follows from the fact that since the two distributions in the above example have mutual informations of $0$ and $\infty$ respectively, their Linfoot correlations are $0$ and $1$ respectively.
\end{proof}

In light of these results, $\popMIC$ can be viewed as a canonical ``minimally smoothed" version of mutual information that is uniformly continuous.

\section{An alternate characterization of $\popMIC$}
In this section we show that $\popMIC$ can be characterized as a supremum over a boundary of the population characteristic matrix defined in Definition~\ref{def:charmatrix}, instead of as a supremum over all the entries of the population characteristic matrix. This accomplishes two goals: first, it will allow us to propose for the first time an algorithm to compute to arbitrary precision the true $\popMIC$ of a random variable with a known pdf. Second, it will be the foundation for the consistency of the two new estimators of $\popMIC$ that we introduce.

\subsection{The boundary of the characteristic matrix}
We begin with the following observation.
\begin{prop}
Let $M$ be a population characteristic matrix. Then for $k \geq \ell$, $M_{k,\ell} \leq M_{k+1, \ell}$.
\end{prop}
\begin{proof}
Let $(X,Y)$ be the random variable in question. Since we can always let a row/column be empty, we know that $I^*((X,Y),k,\ell) \leq I^*((X,Y), k+1, \ell)$. And since $k, k+1 \geq \ell$, we know that $M_{k,\ell} = I^*((X,Y),k,\ell)/\log\ell \leq I^*((X,Y),k+1,\ell)/\log \ell = M_{k+1, \ell}$.
\end{proof}

Since the entries of the characteristic matrix are bounded, the monotone convergence theorem then gives us the following corollary.
\begin{cor}
\label{cor:boundaryIsSup}
Let $M$ be a population characteristic matrix. Then $M_{k,\uparrow} = \lim_{\ell \rightarrow \infty} M_{k,\ell}$ exists, is finite, and equals $\sup_\ell M_{k,\ell}$. The same is true for $M_{\uparrow,\ell}$.
\end{cor}

The above corollary allows us to define the {\em boundary} of the characteristic matrix.

\begin{definition}
Let $M$ be a population characteristic matrix. The {\em boundary} of $M$ is the set
\[ \partial M = \{ M_{k, \uparrow} : 1 < k < \infty \} \bigcup \{ M_{\uparrow, \ell} : 1 < \ell < \infty \} \]
\end{definition}

\subsection{A formula for the boundary of the population characteristic matrix}
The boundary of the population characteristic matrix will be an important object for us. One reason for this is that elements of the boundary can be expressed in terms of a maximization over (one-dimensional) partitions rather than (two-dimensional) grids, the former being much quicker to compute exactly. The proposition below shows this.
\begin{prop}
\label{prop:explicitValueOfBoundary}
Let $M$ be a population characteristic matrix. Then $M_{k, \uparrow}$ equals
\[ \sup_{P \in P(k)} \frac{I(X, Y|_P)}{\log{k}} \]
where $P(k)$ denotes the set of all partitions of $[0,1]$ into at most $k$ pieces.
\end{prop}
\begin{proof}
Define
\[ M_{k,\uparrow}^* = \sup_{P \in P(k)} \frac{I(X, Y|_P)}{\log{k}} \]
We wish to show that $M_{k,\uparrow}^*$ is in fact equal to $M_{k,\uparrow}$. To show that $M_{k,\uparrow} \leq M_{k,\uparrow}^*$, we observe that for every $k$-by-$\ell$ grid $G = (P,Q)$, where $P$ is a partition into rows and $Q$ is a partition into columns, the data processing inequality gives $I((X, Y)|_G) \leq I(X, Y|_P)$. Thus $M_{k,\ell} \leq M_{k,\uparrow}^*$ for $\ell \geq k$, implying that
\[ M_{k,\uparrow} = \lim_{\ell \rightarrow \infty} M_{k,\ell} \leq M_{k,\uparrow}^* \]

It remains to show that $M_{k,\uparrow}^* \leq M_{k,\uparrow}$. To do this, we let $P$ be any partition into $k$ rows, and we define $Q_\ell$ to be an equipartition into $\ell$ columns. We let
\[ M_{k,\ell,P}^* = \frac{I(X|_{Q_\ell}, Y|_P)}{\log{k} } \]
Since $M_{k,\ell,P}^* \leq M_{k,\ell}$ when $\ell \geq k$, we have that for all $P$
\[ \frac{I(X, Y|_P)}{\log{k}} = \lim_{\ell \rightarrow \infty} M_{k,\ell,P}^* \leq \lim_{\ell \rightarrow \infty} M_{k,\ell} = M_{k,\uparrow} \]
which gives that
\[ M_{k,\uparrow}^* = \sup_P \frac{I(X, Y|_P)}{\log{k}} \leq M_{k,\uparrow} \]
as desired.
\end{proof}

\subsection{$\popMIC$ in terms of the boundary of the population characteristic matrix}
We now show that $\popMIC$ is actually the maximum over the boundary of the characteristic matrix. This will mean that to compute $\popMIC$ we need only compute the boundary of the characteristic matrix rather than the entire matrix.
\begin{thm}
\label{thm:MICinTermsOfBoundary}
Let $(X,Y)$ be a random variable. We have
\[ \popMIC(X,Y) = \sup \partial M \]
where $M$ is the population characteristic matrix of $(X, Y)$.
\end{thm}
\begin{proof}
The following argument shows that every entry of $M$ is at most $\sup \partial M$: fix a pair $(k, \ell)$ and notice that either $k \leq \ell$, in which case $M_{k,\ell} \leq M_{k,\uparrow}$, or $\ell \leq k$, in which case $M_{k,\ell} \leq M_{\uparrow, \ell}$. Thus, $\MIC \leq \sup \{M_{\uparrow, \ell}\} \cup \{M_{k, \uparrow}\} = \sup \partial M$.

On the other hand, Corollary~\ref{cor:boundaryIsSup} shows that each element of $\partial M$ is a supremum over some elements of $M$. Therefore, $\sup \partial M$, being a supremum over suprema of elements of $M$, cannot exceed $\sup M = \popMIC$.
\end{proof}

This theorem is the basis for the algorithms introduced in the following three sections.

\section{An exactly computable, consistent estimator of $\popMIC$ using equipartitions}
In the first section, we showed that the statistic $\MIC$ is simply an estimator of $\popMIC$ as defined in this paper. We will now introduce a second estimator of $\popMIC$, which we call $\MICestE$, on the basis of the alternate characterization of $\popMIC$ given in the previous section.

The new estimator $\MICestE$ is analogous to APPROX-MIC, the algorithm given in~\cite{MINE} for heuristically approximating $\MIC$. We therefore first review APPROX-MIC: it used a heuristic for efficiently computing $I^*(D, k, \ell)$ wherein the dimension being partitioned into fewer rows/columns is equipartitioned while the remaining dimension is optimized using dynamic programming. The rationale for this was that since the mutual information is bounded by the marginal entropy along the axis with fewer rows/columns, that axis may as well have its marginal entropy maximized by an equipartition. However, this was a heuristic assumption and was not rigorously justified.

In this section, we use the alternate characterization of $\popMIC$ given above to show that in fact if the dimension with {\em more} rows/columns is equipartitioned while the remaining dimension is optimized, then the resulting statistic is actually a consistent estimator of $\popMIC$ rather than a heuristic approximation of a consistent estimator of $\popMIC$. More specifically, we use the fact that $\popMIC(X,Y) = \sup \partial M(X,Y)$ to say that $\popMIC$ can in fact be realized as the supremum of a modified version of the population characteristic matrix that is easier to compute efficiently. We then observe that the new estimator $\MICestE$ estimates this latter quantity consistently.

As a matter of notation, we first define a version of $I^*$ that equipartitions the dimension with more rows/columns.
\begin{definition}
Let $(X,Y)$ be a random variable. Define
\[ I^* \left( (X,Y), k, [\ell] \right) = \max_{G \in G(k, [\ell])} I \left( (X,Y)|_G \right) \]
where $G(k,[\ell])$ is the set of $k$-by-$\ell$ grids whose y-axis partition is an equipartition of size $\ell$. Define $I^*((X,Y), [k], \ell)$ analogously.

Define $I^{[*]}((X,Y), k, \ell)$ to equal $I^*((X,Y), k, [\ell])$ if $k \leq \ell$ and $I^*((X,Y), [k], \ell)$ otherwise.
\end{definition}

We now prove that if we use $I^{[*]}$ in place of $I^*$ in the population characteristic matrix, then the supremum of the modified matrix is still equal to $\popMIC$.

\begin{thm}
\label{thm:altDef2}
Let $(X,Y)$ be a random variable, and define
\[ [M](X,Y)_{k,\ell} = \frac{I^{[*]}((X,Y), k, \ell)}{\log \min \{k, \ell\}} \]
Then $\sup_{k,\ell} [M](X,Y) = \popMIC(X,Y)$
\end{thm}
\begin{proof}
Because $[M]_{k,\ell}$ is monotonically increasing in $\ell$ for $\ell \geq k$ but also bounded, $[M]$ can be shown to have a boundary in the same sense of $M$. What does this boundary look like? As we increase $\ell$ while holding $k$ fixed the axis being equipartitioned will have a finer and finer equipartition, such that in the limit that axis will become a continuous random variable. Thus, the same argument used in Proposition~\ref{prop:explicitValueOfBoundary}, shows that 
\[ [M]_{k,\uparrow} = \sup_{P \in P(k)} \frac{I(X, Y|_P)}{\log{k}} = M_{k, \uparrow} \]
Thus, $\partial [M] = \partial M$.

Finally, because every sequence of the form $[M]_{k,2}, [M]_{k, 3}, \ldots$ or $[M]_{2,\ell}, [M]_{3, \ell}, \ldots$ eventually becomes non-decreasing, the same argument given in Theorem~\ref{thm:MICinTermsOfBoundary} shows that $\sup [M] = \sup \partial [M]$. Therefore,
\[ \popMIC = \sup M = \sup \partial M = \sup \partial [M] = \sup [M] \]
as desired.
\end{proof}

We can now define $\MICestE$ to be the estimator of $\sup [M]$ that is analogous to $\MIC$.
\begin{definition}
Let $D \subset [0,1] \times [0,1]$ be a set of ordered pairs. Given a function $B : \Z^+ \rightarrow \Z^+$, we define
\[ \widehat{[M]}_B(D)_{k,\ell} = \left\{
        \begin{array}{ll}
            \frac{I^{[*]}(D, k, \ell)}{ \log \min \{k, \ell \}} & \quad k\ell \leq B(|D|) \\
            0 & \quad k\ell > B(|D|)
        \end{array}
    \right. \]
\end{definition}

\begin{definition}
Let $D \subset [0,1] \times [0,1]$ be a set of ordered pairs, and let $B : \Z^+ \rightarrow \Z^+$. We define
\[ \MICestE_{,B}(D) = \max \widehat{[M]}_B(D) \]
\end{definition}

Since each entry in $\widehat{[M]}_B(D)$ is computed by considering a subset of the grids considered in the computation of $\hat{M}_B(D)$, the same argument that showed the consistency of $\MIC$ shows that $\MICestE$ is a consistent estimator of $\popMIC$ as well. Thus, we obtain

\begin{thm}
The statistic $\MICestE_{,B}$ is a consistent estimator of $\popMIC$ provided $\omega(1) < B(n) \leq O(n^{1-\ep})$ for $\ep > 0$.
\end{thm}

Both $\MIC$ and $\MICestE$ are consistent estimators of $\popMIC$. The difference between them is that while $\MIC$ could only be computed efficiently by a heuristic approximation, $\MICestE$ can be computed exactly and efficiently by a trivial adaptation of the APPROX-MIC algorithm.

The performance of $\MICestE$ in terms of bias/variance, equitability, power against independence, and runtime is compared to that of $\MIC$ as well as other methods in the forthcoming companion paper.

\section{An algorithm for computing the $\popMIC$ of a given density function to arbitrary precision}
In this section we give an algorithm for computing the $\popMIC$ of a given probability density function to arbitrary precision. Our ability to do so stems from the observation that Theorem~\ref{thm:MICinTermsOfBoundary}, together with Proposition \ref{prop:explicitValueOfBoundary}, yields the following corollary.
\begin{cor}
\label{cor:altDef1}
Let $(X,Y)$ be a random variable taking values in $[0,1]^2$, and let $\mathbb{P}$ be the set of finite partitions of $[0,1]$. Then
\[ \popMIC(X,Y) = \sup \left\{ \frac{I(X, Y|_P)}{\log |P|} : P \in \mathbb{P} \right\} \bigcup \left\{ \frac{I(X|_P, Y)}{\log |P|} : P \in \mathbb{P} \right\} \]
where $|P|$ is the number of bins in the partition $P$.
\end{cor}
The expressions in the above corollary involve maximization only over one-dimensional partitions rather than two-dimensional grids. This type of maximization can be done efficiently using dynamic programming, and this is the algorithm we will propose.

In addition to allowing us to reason rigorously about $\popMIC$ in the large-sample limit, this algorithm will also provide the basis for a second new estimator of $\popMIC$ that works by estimating the density of the distribution that gave rise to the data and then computing the true $\popMIC$ of that density.

We present this argument formally in the following theorem.
\begin{thm}
Given the pdf of a random variable $(X,Y)$, $M_{k, \uparrow}$ and $M_{\uparrow, \ell}$ are both computable to within an additive error of $O(k \ep^{0.999})+E$ in time $O(kT/\ep)$, where $T$ is the time required to numerically compute the mutual information of a continuous distribution to within an error of $E$.
\end{thm}
\begin{proof}
We prove the claim for $M_{k, \uparrow}$. For $0 < \ep, \ep' < 1$, let $\Gamma = (\Pi, \Psi)$ be a grid consisting of an equipartition $\Pi$ into $1/\ep$ rows, and an equipartition $\Psi$ into $1/\ep'$ columns. The dynamic programming algorithm presented in~\cite{MINE} finds, in time $O(kT/\ep)$, a partition $P_{DP}$ into rows such that $I \left(X_\Psi,Y|_{P_{DP}} \right)$ is maximized subject to $P_{DP} \subset \Pi$. We must bound the difference between this and $I \left(X, Y|_P \right)$, where $P$ is an optimal partition into rows. We will do so in two steps: first we will bound the quantity
\[ \left| I \left( X|_\Psi, Y|_{P_{DP}} \right) - I \left( X|_\Psi, Y|_P \right) \right| \]
we will then bound
\[ \left| I \left( X|_\Psi, Y|_P \right) - I \left( X, Y|_P \right) \right| \]

We will prove the first bound by taking the grid $(P, \Psi)$ and showing that there exists some $\Pi' \subset \Pi$ such that the mutual information achieved with $(\Pi', \Psi)$ is close to that achieved with $(P, \Psi)$. Since $\Pi' \subset \Pi$ gives us that $I\left( X_\Psi, Y_{\Pi'} \right) \leq I\left( X_\Psi, Y_{P_{DP}} \right)$, we may then conclude that $\left| I \left( X|_\Psi, Y|_{P_{DP}} \right) - I \left( X|_\Psi, Y|_P \right) \right|$ is small.

We construct $\Pi'$ using Lemma~\ref{lem:boundMIOfCloseGrids} from the appendix with $G = (\Psi, P)$. The lemma shows that we can define $P'$ to be the partition into rows that results from adding to $P$ the two lines in $\Pi$ that surround each line of $P$ that is not in $\Pi$. If we define $P'$ in this way, we get that
\[ \left| I \left( X|_\Psi, Y|_{P'} \right) - I \left( X|_\Psi, Y|_P \right) \right| \leq 2 \left( H_b(\delta) + \delta \right) \]
where $\delta$ is the total probability mass of $(X,Y)$ lying in cells of $(\Psi, P)$ that are not contained in individual cells of $(\Psi, \Pi)$. We know, since there are at most $k-1$ horizontal lines in $(\Psi, P)$, that this number is at most $(k-1)\ep$. We are not quite done though: the new grid $(\Psi, P')$ is indeed a sub-grid of $\Gamma$, but it may now contain up to $2k-1$ rows. This means that we may have to remove up to $k$ lines from it without losing too much mutual information. Fortunately though, for every pair of new lines in $P'$, we can arbitrarily remove one of the lines in the pair. Lemma~\ref{lem:boundMIOfCloseGrids_Subgrids} shows that each time we do so we lose at most $\pi H_b(\ep/\pi)$ in mutual information, where $\pi$ is the total probability mass of the new merged row. Applying Lemma~\ref{lem:boundWeightedAverageOfEntropy} then shows that this quantity is at most $H_b(\ep)$. This leaves us with
\[ \left| I \left( X|_\Psi, Y|_{P'} \right) - I \left( X|_\Psi, Y|_P \right) \right| \leq 2 \left( H_b((k-1)\ep) + (k-1)\ep \right) + k H_b(\ep) \]
from which the result is obtained by observing that $H_b(p) \leq O(p^{1-\alpha})$ for any $\alpha > 0$.
\end{proof}

\begin{rmk}
We do not explore here the numerical integration associated with the above theorem, since the error introduced by choice of method is independent of the algorithm being proposed. However, standard numerical integration methods can be used to make this error arbitrarily small with an understood complexity tradeoff (see, e.g., \cite{stoer1980numerical}). In the upcoming companion paper by the same authors, we discuss the details of implementing this algorithm, including the numerical integration step.
\end{rmk}

We have shown so far that it is possible to efficiently approximate $M_{k, \uparrow}$ and $M_{\uparrow, \ell}$. A simple corollary of this is that we can also approximate the $\popMIC$ of a pdf to arbitrary precision.

\begin{cor}
Given the pdf of a random variable $(X,Y)$ and an error threshold $\ep$, it is possible to compute $\popMIC(X,Y)$ to within additive error $\ep$.
\end{cor}
\begin{proof}
Let
\[ \partial M_s = \left\{ M_{k, \uparrow} : k \leq s \right\} \bigcup \left\{ M_{\uparrow, \ell} : \ell \leq s \right\} \]
The corollary follows from the fact that the sequence $\sup \partial M_s$ converges to $\sup \partial M$.
\end{proof}

\subsection{An efficiently computable, consistent estimator of $\popMIC$ using density estimation}

The fact that we can compute the $\popMIC$ of a pdf to arbitrary precision, combined with the continuity of $\popMIC$ proven in a previous section, yields a second new estimator of $\popMIC$ that is consistent: given a finite sample, we can estimate the density that gave rise to it and then simply compute the $\popMIC$ of that density. This will give a consistent estimator of $\popMIC$ if the density estimator we use is consistent. We call this estimator $\MICestD$.

Again, we let $\mathcal{P}(S)$ be the set of probability density functions whose support is contained in $S$, equipped with the $L^1$ norm.

\begin{thm}
Let $F$ be a consistent density estimator. Given a finite sample $D$ from a distribution $\mcZ$ of a random variable $(X,Y)$, let $\popMIC(F(D))$ be the $\popMIC$ of the random variable whose distribution is $F(D)$. Then $\MICestD_{,F}(D) = \popMIC(F(D))$ is a consistent estimator of $\popMIC(X,Y)$.
\end{thm}
\begin{proof}
Follows from the continuity of $\popMIC : \mathcal{P}([0,1]^2) \rightarrow \R$, together with the continuous mapping theorem (if $X_i \rightarrow X$ in probability and $f$ is continuous, then $f(X_i) \rightarrow f(X)$ in probability as well).
\end{proof}

\section{Conclusion}
In this paper, we formalized and developed the theory behind both equitability and the maximal information coefficient. We first defined equitability in terms of a statistic $\hvphi$ and a model $\Q$ of ``standard relationships" on which we specify a property of interest $\Phi$ that reflects our notion of relationship strength. The equitability of a statistic $\hvphi$ is the extent to which, for any $\mcZ \in \Q$, knowing $\hvphi(Z^*)$ gives us good bounds on $\Phi(Z)$, where $Z^*$ denotes a sample from $Z$. We showed that this property can be stated in terms of power against a specific set of null hypotheses corresponding to different threshold values of the property of interest $\Phi$. In particular, when $\Phi$ is chosen such that $\Phi = 0$ corresponds to statistical independence, equitability is a generalization of the notion of power against statistical independence.

Having defined equitability, we then turned our attention to the maximal information coefficient (MIC). The original paper about $\MIC$ defined it as a statistic. In this paper, however, we defined a function of distributions called $\popMIC$ and then proved that $\MIC$ can be understood to be a consistent estimator of that quantity. We then went on to prove first that $\popMIC$  is continuous when considered as a function of probability densities, and second that it can be equivalently characterized in simpler terms that open up new avenues for computing it. These results together led us to define two new efficiently computable, consistent estimators of $\popMIC$. (These objects are summarized in Table~\ref{tab:objectsDefined} below.) In addition to these estimators, we also described an algorithm for computing to arbitrary precision the true $\popMIC$ of a given probability density function that will be useful in analyzing both the behavior of $\popMIC$ and of these estimators.

\begin{table}[h!]
\centering
	\begin{tabular}{p{0.5in}p{3in}p{0.9in}}
	\textbf{Object} &
	\textbf{Description} &
	\textbf{Defined in} \\ \hline \\
	$\MIC$ &
	Original statistic &
	\cite{MINE} \vspace{1\bigskipamount} \\ \hline \\
	 
	$\popMIC$ &
	Population value of $\MIC$ &
	here \vspace{1\bigskipamount} \\ \hline \\
	
	$\MICestE$ &
	Estimator of $\popMIC$ via equipartitions &
	here \vspace{1\bigskipamount} \\ \hline \\
	
	$\MICestD$ &
	Estimator of $\popMIC$ via density estimation &
	here \vspace{1\bigskipamount} \\ \hline \\
\end{tabular}
\caption{Currently defined statistics and estimands related to $\MIC$.} \label{tab:objectsDefined}
\end{table}

\subsection{Open questions}
Our results leave several open questions, both theoretical and practical. On the theoretical side, it would be valuable to have more than just the consistency of the estimators we have introduced. Closed-form expressions for their bias and variance, while likely difficult to compute, would be useful for understanding the tradeoffs among these (and other) estimators. Second, now that $\MIC$ is defined as a property of distributions rather than a statistic, it would make sense to obtain closed-form expressions for the $\MIC$ of some canonical family of distributions (e.g. Gaussians). This would contribute to a better understanding of what is captured by $\MIC$ and how it relates to other well understood notions of dependence. Finally, theoretical guarantees about the equitability of $\MIC$ in the infinite-data limit on some set $\Q$ would be highly desirable. Perhaps the alternative characterization proven in this paper will be a useful first step toward these results.

More experimental open questions abound as well. For instance, how do each of these estimators compare to each other and to the state of the art on common finite sample sizes, both in terms of bias and variance, and in terms of equitability and power? How fast are they relative to each other and other methods? A forthcoming companion paper, which goes into the details of the implementation of the algorithmic ideas introduced here, addresses these experimental questions and shows that the ideas introduced here in fact lead to significant improvement in all four of these realms.

\section{Acknowledgments}
The authors would like to acknowledge R Adams, E Airoldi, H Finucane, A Gelman, J Huggins, J Mueller, and R Tibshirani for constructive conversations and useful feedback.

\newpage

\appendix
\label{appendix:technical}

\section{Lemmas about entropy and mutual information}
\begin{lemma}
\label{lem:boundEntropy}
Let $\Pi$ and $\Psi$ be random variables distributed over a discrete set of states $\Gamma$, and let $(\pi_i)$ and $(\psi_i)$ be their respective distributions. Let $P = f(\Pi)$ and $Q = f(\Psi)$ for some function $f$ whose image is of size $B$. Define
\[ \ep_i = \frac{\psi_i - \pi_i}{\pi_i} \]
Then for every $0 < a < 1$ there exists some $A > 0$ such that
\[ \left| H(Q) - H(P) \right| \leq \left( \log B \right) A \sum_i |\ep_i| \]
when $|\ep_i| \leq 1 - a$ for all $i$.
\end{lemma}
\begin{proof}
We prove the claim with entropy measured in nats. A re-scaling will then give the general result.

Let $(p_i)$ and $(q_i)$ be the distributions of $P$ and $Q$ respectively, and define
\[ e_i = \frac{q_i - p_i}{p_i} \]
analogously to $\ep_i$. Before proceeding, we observe that
\[ e_i = \sum_{j \in f^{-1}(i)} \frac{\pi_j}{p_i} \ep_j \]

We now proceed with the argument. We have from \cite{roulston1999estimating} that
\begin{eqnarray}
\left| H(Q) - H(P) \right| &\leq& \left| \sum_i \left( e_i p_i (1 + \ln p_i) + \frac{1}{2} e_i^2 p_i + \O{e_i^3} \right) \right| \\
	&\leq& \left| \sum_i e_i p_i \right| + \left| \sum_i e_i p_i \ln p_i \right| + \frac{1}{2} \left| \sum_i e_i^2 p_i \right| + \left| \sum_i \O{e_i^3} \right| \\
	&=& \left| \sum_i e_i p_i \ln p_i \right| + \frac{1}{2} \sum_i e_i^2 p_i + \left| \sum_i \O{e_i^3} \right| \label{eq:toBound}
\end{eqnarray}
where the final equality is because $\sum_i e_i p_i = \sum_i q_i - \sum_i p_i = 0$. We will proceed by bounding each of the terms in Equation~\ref{eq:toBound} separately.

To bound the first term, we write
\[ \left| \sum_i e_i p_i \ln p_i \right| \leq -\sum_i | e_i | p_i \ln p_i\]
We then note that $- \sum_i p_i \ln p_i \leq \ln B$, and since each of the summands has the same sign this means that $-p_i \ln p_i \leq \ln B$. We also observe that
\[ \left| e_i \right| \leq \left| \sum_{j \in f^{-1}(i)} \frac{\pi_j}{p_i} \ep_j \right| \leq \sum_j \frac{\pi_j}{p_i} \left| \ep_j \right| \leq \sum_j |\ep_j| \]
since $\pi_j/p_i \leq 1$. Together, these two facts give
\begin{eqnarray*}
- \sum_i |e_i| p_i \ln p_i &\leq& (\ln B) \sum_i |e_i| \\
	&\leq& (\ln B) \sum_i |\ep_i|
\end{eqnarray*}
The second inequality is because 

To bound the second term, we use the fact that $p_i \leq 1$ for all $i$, and so
\[ \sum_i e_i^2 p_i \leq \sum_i e_i^2 \]
We then write
\begin{eqnarray*}
\sum_i e_i^2 &=& \sum_i \left( \sum_{j \in f^{-1}(i)} \frac{\pi_j}{p_i} \ep_j \right)^2 \\
	&\leq& \sum_i \sum_{j \in f^{-1}(i)} \frac{\pi_j}{p_i} \ep_j^2 \\
	&\leq& \sum_j \ep_j^2 \\
	&=& \sum_j \O{\left| \ep_j \right|}
\end{eqnarray*}
where the second line is a consequence of the convexity of $f(x) = x^2$ and the third line is because the sets $f^{-1}(i)$ partition $\Gamma$.

To bound the third term, we write
\[ \left| \sum_i \O{e_i^3} \right| \leq \sum_i \O{|e_i|^3} \]
and then proceed as we did with the second term, using the fact that $f(x) = x^3$ is convex for $x \geq 0$. This gives
\[ \sum_i \O{|e_i|^3} \leq \sum_i \O{|\ep_i|^3} = \sum_i \O{|\ep_i|} \]
completing the proof.
\end{proof}

In the lemma below, we bound the change in mutual information between two grids where one is a sub-grid of the other.
\begin{lemma}
\label{lem:boundMIOfCloseGrids_Subgrids}
Let $G$ and $G'$ be two grids with the property that $G$ can be obtained from $G'$ by merging two adjacent columns of $G'$. Let $Z' = (X', Y')$ be a random variable distributed over the cells of $G'$, and let $Z = f(Z') = (X,Y)$ be the corresponding random variable induced over the cells of $G$. Then
\[ \left| I(Z') - I(Z) \right| \leq \pi H_b(\nu/\pi) \]
where $H_b$ denotes the binary entropy function, $\nu$ is the probability mass contained in one of the columns of $G'$ that was merged, and $\pi$ is the probability mass in the merged column in $G$.
\end{lemma}
\begin{proof}
We use the general fact that $I(A,B) = H(B) - H(B | A)$. In particular, since $Y = Y'$, we have
\[ \left| I(Z') - I(Z) \right| = \left| H(Y | X) - H(Y' | X') \right| \]
Suppose without loss of generality that the first and second columns of $G'$ were merged to form the first column of $G$. Since the distributions of $Z$ and $Z'$ are identical on all the columns of $G$ unaffected by the merge, the expression above equals
\[ \left| \Pr{X = 1}H(Y | X = 1) - \left( \Pr{X' = 1}H(Y' | X' = 1) + \Pr{X' = 2}H(Y' | X' = 2) \right)  \right| \]
Defining the random variable $(A,B) = (X', Y') | X' \in \{1,2\}$ and noting that $B = Y | X = 1$ as random variables allows us to re-write this as
\begin{eqnarray*}
&& \Pr{X = 1} \left| H(B) - H(B | A) \right| \\
	&=& \Pr{X = 1} I(A, B) \\
	&\leq& \Pr{X = 1} H(A) \\
	&=& \pi H_b( \nu / \pi)
\end{eqnarray*}
completing the proof.
\end{proof}

\begin{lemma}
\label{lem:boundWeightedAverageOfEntropy}
Let $\{ w_i \} \subset [0,1]$ be a set of size $n$ with $\sum_i w_i \leq 1$, and let $\{u_i\}$ be a set of $n$ non-negative numbers satisfying $\sum_i u_i = a$ and $u_i \leq w_i$. Then
\[ \sum_{i=1}^n w_i H_b \left( \frac{u_i}{w_i} \right) \leq H_b \left( a \right) \]
\end{lemma}
\begin{proof}
Consider the random variable $X$ taking values in $\{0, \ldots, n\}$ that equals 0 with probability $1-\sum_i w_i$ and equals $i$ with probability $w_i$ for $0 < i \leq n$. Define the random variable $Y$ taking values in $\{0,1\}$ by
\[ \Pr{Y = 0 | X = i} = \left\{
        \begin{array}{ll}
            0 & \quad i = 0 \\
            u_i / w_i & \quad 0 < i \leq n
        \end{array}
    \right.
 \]
The function we wish to bound equals $H(Y | X) \leq H(Y)$. We therefore observe that
\[ \sum_{i=1}^n w_i H_b \left( \frac{u_i}{w_i} \right) \leq H(Y) \]
The result follows from the observation that
\[ \Pr{Y = 0} = \sum_i \Pr{X = i} \frac{u_i}{w_i} = \sum_i u_i \leq a \]
\end{proof}

The following lemma uses Lemma~\ref{lem:boundMIOfCloseGrids_Subgrids} to bound the change in mutual information between two close grids, but without the restriction that one grid be a sub-grid of the other.
\begin{lemma}
\label{lem:boundMIOfCloseGrids}
Let $G$ and $\Gamma$ be two grids, and refer to any horizontal or vertical line of $G$ that is not present in $\Gamma$ as a {\em dissonant} line. Let $G'$ be the grid that results from adding to $G$ the two lines in $\Gamma$ that surround each dissonant line of $G$, and then removing all the dissonant lines from $G$. For every joint distribution $(X,Y)$, we have
\[ \left| I\left((X,Y)|_G\right) - I\left( (X,Y)|_{G^*} \right) \right| \leq 2 \left( H_b(\delta) + \delta \right) \]
where $\delta$ is the total probability mass of $(X,Y)|_\Gamma$ falling in cells of $\Gamma$ that are not contained in individual cells of $G$.
\end{lemma}
\begin{proof}
Let $G''$ denote the union of the lines in $G'$ and $G$. We first assume without loss of generality that the only dissonant line $l_1$ in $G$ is the vertical line separating the first two columns of $G$, and that neither of the new lines in $G'$ are in $G$. Suppose $l_1$ is in the $i$-th column of $\Gamma$, and let $\pi_1^L$ and $\pi_1^R$ be the probability masses in the $i$-th column of $(X,Y)|_\Gamma$ lying to the left and to the right of $l_1$ respectively. Let $p_1^L$ and $p_1^R$ represent the probability mass in the columns of $(X,Y)|_G$ to the left and to the right of $l_1$ respectively. Then two successive applications of Lemma~\ref{lem:boundMIOfCloseGrids_Subgrids}, one for each of the two new lines, show that
\[
\Delta^{(X,Y)}(G, G'') \leq p_1^L H_b \left( \frac{\pi_1^L}{p_1^L} \right) + p_1^R H_b \left( \frac{\pi_1^R}{p_1^R} \right)
\]
where $H_b$ denotes the binary entropy function.

To get from $G''$ to $G'$, we again apply Lemma~\ref{lem:boundMIOfCloseGrids_Subgrids} to obtain
\[ \Delta^{(X,Y)}(G'', G')
\leq \left( \pi_1^L + \pi_1^R \right) H_b \left( \frac{\pi_1^L}{\pi_1^L + \pi_1^R} \right) \leq \left( \pi_1^L + \pi_1^R \right) \]
Thus,
\[ \Delta^{(X,Y)}(G, G')
\leq p_1^L H_b \left( \frac{\pi_1^L}{p_1^L} \right) + p_1^R H_b \left( \frac{\pi_1^R}{p_1^R} \right) + \left( \pi_1^L + \pi_1^R \right) \]

We next treat the more general case in which $G$ contains more than one vertical dissonant line but no horizontal dissonant lines. In this case, since $G$ has $\ell$ columns, there are at most $\ell-1$ vertical dissonant lines $l_1, l_2, \ldots, l_{\ell-1}$. Applying the above procedure for each of these results in a grid $G'$ with at most $2\ell$ columns such that
\begin{eqnarray*}
\Delta^{(X,Y)}(G, G')
&\leq& \sum_{i=1}^{\ell-1} \left( p_i^L H_b \left( \frac{\pi_i^L}{p_i^L} \right) + p_i^R H_b \left( \frac{\pi_i^R}{p_i^R} \right) \right)
	+ \sum_{i=1}^{\ell-1} \left( \pi_i^L + \pi_i^R \right) \\
&\leq& H_b \left( \sum_{i=1}^{\ell-1} \left( \pi_i^L + \pi_i^R \right) \right)
	+ \sum_{i=1}^{\ell-1} \left( \pi_i^L + \pi_i^R \right) \\
&=& H_b \left( \delta_C \right) + \delta_C
\end{eqnarray*}
where the second inequality follows from application of Lemma~\ref{lem:boundWeightedAverageOfEntropy}, and $\delta_c$ is the total probability mass in all columns of $(X,Y)|_\Gamma$ that contain a dissonant vertical line of $G$. An analogous argument applied to the dissonant horizontal lines of $G$ then shows that in the entirely general case we have a grid $G'$ with at most $2\ell$ columns and $2k$ rows such that
\begin{eqnarray*}
\Delta^{(X,Y)}(G, G') &\leq& H_b \left(\delta_C \right) + H_b \left(\delta_R \right) + \delta_C + \delta_R \\
	&\leq& 2H_b(\delta) + 2\delta
\end{eqnarray*}
where $\delta_R$ is defined analogously to $\delta_C$, we have observed that $\delta_C + \delta_R \leq 2\delta$, and $H_b(\delta_C), H_b(\delta_R) \leq H_b(\delta)$ since $\delta_R, \delta_C \leq \delta \leq 1/2$.
\end{proof}

\begin{lemma}
Let $X$ be a random variable distributed over $k$ states, with $\Pr{X = x} = p_x$. Let $\alpha_x \geq 0$ be such that $\sum \alpha_x = \delta$, and define the random variable $X'$ by $\Pr{X' = x} = (p_x + \alpha_x)/(1 + \delta)$. We have
\[ \left| H(X') - H(X) \right| \leq H_b(\delta) + \delta \log k \]
\end{lemma}
\begin{proof}
Define a new random variable $Z$ by
\[ \Pr{Z = 0 | X' = x} = \frac{p_x}{p_x+\alpha_x}, \Pr{Z = 1 | X' = x} = \frac{\alpha_x}{p_x + \alpha_x} \]
We will use the fact that $H(X' | Z = 0) = H(X)$ to achieve our bound.

To upper bound $H(X') - H(X)$, we write
\begin{eqnarray*}
H(X') - H(X) &\leq& H(X', Z) - H(X)  \\
	&=& H(Z) + \Pr{Z = 0} H(X' | Z = 0) + \Pr{Z = 1} H(X' | Z = 1) - H(X) \\
	&\leq& H_b(\delta) + (1-\delta) H(X) + \delta H(X' | Z = 1) - H(X) \\
	&=& H_b(\delta) - \delta H(X) + \delta \log k \\
	&\leq& H_b(\delta) + \delta \log k
\end{eqnarray*}
where in the fourth line we have used that $H(X' | Z = 1) \leq \log k$.

To upper bound $H(X) - H(X')$, we write
\begin{eqnarray*}
H(X') + H(Z) &\geq& H(X', Z) \\
	&\geq& \Pr{Z=0} H(X' | Z = 0) \\
	&=& (1-\delta) H(X)
\end{eqnarray*}
which yields
\[ H(X') \geq (1-\delta) H(X) - H_b(\delta) \]
since $H(Z) = H_b(\delta)$. Thus, we have
\[ H(X) - H(X') \leq \delta H(X) + H_b(\delta) \leq \delta \log k + H_b(\delta) \]
\end{proof}

\begin{lemma}
Let $X$ be a random variable distributed over $k$ states, with $\Pr{X = x} = p_x$. Let $\alpha_x \leq 0$ be such that $\sum |\alpha_x| = \delta$, and define the random variable $X'$ by $\Pr{X' = x} = (p_x + \alpha_x)/(1 - \delta)$. We have
\[ \left| H(X') - H(X) \right| \leq H_b(\frac{\delta}{1-\delta}) + \frac{\delta}{1-\delta} \log k \]
In particular, when $\delta \leq 1/3$ we have
\[ \left| H(X') - H(X) \right| \leq H_b(2\delta) + 2\delta \log k \]
\end{lemma}
\begin{proof}
We observe that we can get from $X'$ to $X$ by adding $\delta / (1-\delta)$ probability mass and re-scaling. The previous lemma then gives the result.
\end{proof}

\begin{lemma}
\label{lem:entropyBoundForChangedMass}
Let $X$ be a random variable distributed over $k$ states, with $\Pr{X = x} = p_x$. Let $\alpha_x$ be such that $\sum |\alpha_x| = \delta$, and define the random variable $X'$ by $\Pr{X' = x} = (p_x + \alpha_x)/(1 - \sum \alpha_x)$. If $\delta \leq 1/3$, we have
\[ \left| H(X') - H(X) \right| \leq 2H_b(2\delta) + 3\delta \log k \]
\end{lemma}
\begin{proof}
Let $\delta_+$ be the total magnitude of all the positive $\alpha_x$, and let $\delta_-$ be the total magnitude of all the negative $\alpha_x$. We first add all the mass we're going to add, and apply the first of the previous two lemmas. Then we remove all the mass we are going to remove, and apply the second of the two previous lemmas. This yields a bound of
\begin{eqnarray*}
&&	H_b(\delta_+) + \delta_+ \log k + H_b\left(2\frac{\delta_-}{1 + \delta_+} \right) + 2 \frac{\delta_-}{1 + \delta_+} \log k \\
&\leq& H_b(\delta_+) + \delta_+ \log k + H_b(2 \delta_-) + 2\delta_- \log k \\
&\leq& H_b(2\delta) + \delta \log k + H_b(2 \delta) + 2 \delta \log k \\
&\leq& 2H_b(2 \delta) + 3 \delta \log k
\end{eqnarray*}
where the first inequality is because $1 + \delta_+ \leq 1 + \delta < 2$ and $2 \delta_- \leq 2 \delta \leq 1/2$, and the second inequality is because $\delta_+ \leq \delta < 2\delta \leq 1/2$.
\end{proof}

\newpage
\bibliographystyle{ieeetr}
\bibliography{\pathToCommon/References}

\end{document}